%% file: extension.tex
\documentclass[journal,twoside,web]{ieeecolor} 

\usepackage{generic}
\usepackage{graphicx}
\usepackage{amsmath} 
\usepackage{amssymb}
\usepackage{amsfonts}
\usepackage{mathrsfs}
\usepackage{mathtools}
\usepackage{cite}
\usepackage[bbgreekl]{mathbbol}
\usepackage{multirow}
\usepackage{makecell}
\usepackage{booktabs}
\usepackage{caption}
\usepackage{xpatch}
\usepackage{float}
\usepackage{tikz}
\usepackage{pgfplots} 
\usepackage{pgfgantt}
\usepackage{pdflscape}
\usepackage{xcolor}
\usepackage[table]{xcolor}
\usepackage{subcaption}

\pgfplotsset{compat=newest}

\let\labelindent\relax
\usepackage{enumitem}

\usepackage{amsthm}

\DeclareSymbolFontAlphabet{\mathbb}{AMSb}
\DeclareSymbolFontAlphabet{\mathbbl}{bbold}

\DeclareMathOperator{\diag}{diag}
\DeclareMathOperator{\rank}{rank}
\DeclareMathOperator{\blkdiag}{blkdiag}

\DeclareMathOperator{\diff}{d}

\DeclareMathOperator*{\argmin}{argmin}

\DeclareMathOperator{\Ima}{Im}
\DeclareMathOperator{\vspan}{span}

\newcommand{\mc}{\mathcal}
\newcommand{\ddt}{\tfrac{\diff}{\diff \!t}}
\newcommand{\norm}[1]{\left \lVert #1 \right \rVert}

\makeatletter
\xpatchcmd{\@thm}{\thm@headpunct{.}}{\thm@headpunct{}}{}{}
\makeatother

\definecolor{lightgray}{gray}{0.9}

\newtheorem{theorem}{Theorem}
\newtheorem{lemma}{Lemma}
\newtheorem{proposition}{Proposition}
\newtheorem{assumption}{Assumption}
\newtheorem{remark}{Remark}
\newtheorem{definition}{Definition}

\newtheorem{corollary}{Corollary}

\title{\LARGE \bf  Stability and convergence of multi-converter systems using projection-free power-limiting droop control}

\author{Amirhossein Iraniparast  and Dominic Gro\ss{} \thanks{This work was supported in part by the National Science Foundation under Grant No. 2143188. A. Iraniparast and D. Gro\ss{} are with the Department of Electrical and Computer Engineering at the University of Wisconsin-Madison, USA; e-mail: iraniparast@wisc.edu, dominic.gross@wisc.edu}}

\begin{document}
\maketitle

\begin{abstract} 
    In this paper, we propose a projection-free power-limiting droop control for grid-connected power electronics and an associated constrained flow problem. In contrast to projection-based power-limiting droop control, the novel projection-free power-limiting droop control results in networked dynamics that are semi-globally exponentially stable with respect to the set of optimizers of the constrained flow problem. Under a change to edge coordinates, the overall networked dynamics arising from projection-free power-limiting droop control coincide with the projection-free primal-dual dynamics associated with an augmented Lagrangian of the constrained flow problem. Leveraging this result, we (i) provide a bound on the convergence rate of the projection-free networked dynamics, (ii) propose a tuning method for controller parameters to improve the bound on the convergence rate, and (iii) analyze the relationship of the bound on the convergence rate and connectivity of the network. Finally, the analytical results are illustrated using an Electromagnetic transient (EMT) simulation.
\end{abstract}
\section{Introduction}
The ongoing shift from synchronous machine-based power generation towards power electronics-interfaced generation and energy storage results in significant changes to power system frequency dynamics. Specifically, grid-connected power electronics differ from conventional synchronous generators in terms of their fast response (i.e., milliseconds to seconds) and resource constraints (e.g. power and current limits). Accordingly, incorporating renewable generation resources into large-scale power system challenges standard operating and control paradigm and jeopardizes system stability~\cite{KJZ+2017,MDHHV2018}. For instance, stability analysis of emerging power systems crucially requires considering the constraints of power converters and renewable generation resources such as power limit.

Today, most renewables are interfaced by dc/ac voltage source converters (VSC) use so-called grid-following control. This control paradigm requires a stable and slowly changing ac voltage (i.e., magnitude and frequency) and jeopardizes grid stability when disturbance occur~\cite{KH2024}. Since grid-following explicitly controls the converter current/power, incorporating power limits is straightforward. In contrast, grid-forming converters, that are commonly envisioned to be the cornerstone of future power systems, impose stable and self-synchronizing ac voltage dynamics at their grid terminals. Although prevalent grid-forming controls including droop control\cite{CDA1993}, virtual synchronous machine control (VSM)\cite{DSF2015}, and dispatchable virtual oscillator control (dVOC)\cite{GCBD2019} have been investigated in detail constraints are not accounted for in their analysis\cite{DB2012,SOARS2014,SGRS2013,DS2014,SGCD2021}.

However, from a practical point of view, resource and converter constraints are a significant concern. The majority of works on grid-forming control under constraints in the application oriented literature has focused on current-limiting (see \cite{BCL2024} for a recent survey). To address power-limiting, power-limiting droop control combines conventional droop control with proportional-integral limiters that activate when the converter reaches its power limit~\cite{DLK2019}. While asymptotic stability of power-limiting droop control has been established~\cite{IG2025}, this result does not quantify the impact of controller gains or network parameters.

Continuous-time primal-dual dynamics~\cite{CMC2016} have been widely used to study system-level controls arising from equality constrained optimization problems such as secondary control and economic dispatch of multi-machine systems~\cite{LZC2016} and power flow control~\cite{CDA2020}. In contrast, directly applying primal-dual dynamics to primary control design results in controls with communication requirements that are infeasible in practice.

This work leverages projection-free primal-dual dynamics~\cite{QL2019} in edge coordinates to develop a novel grid-forming power-limiting primary control for grid-connected power electronics and analyze the resulting multi-converter system frequency dynamics. Our novel projection-free power-limiting droop control is distinct from power-limiting droop control~\cite{DLK2019,IG2025} and enables rigorous bounds on the convergence rate of the multi-converter system frequency to the optimal solution of an associated constrained dispatch problem. 

To characterize the steady-states of these networked dynamics, a generic constrained network flow problem can be formulated whose primal-dual dynamics cannot be implemented using only local information~\cite{IG2025}. In contrast, applying primal-dual dynamics in edge coordinates~\cite{ZM2011} results in a fully decentralized converter control that is asymptotically stable with respect to Karush-Kuhn-Tucker (KKT) points of the constrained flow problem in edge coordinates.

This enables representing the networked dynamics as primal-dual dynamics and apply well-known stability results~\cite{CMC2016}. Notably, the networked dynamics in nodal coordinates are globally asymptotically stable with respect to the set of optimizers of its associated constrained network problem in nodal coordinates~\cite{IG2025}. This result directly establishes frequency stability and synchronization of networks of converters using power-limiting grid-forming droop control. In addition, upon convergence, the converters exhibit power-sharing properties similar to power-sharing in unconstrained droop control~\cite{SDB2013}.

However, the discontinuity of power-limiting droop control hinders convergence analysis and no convergence rate is provided in~\cite{IG2025}. From a practical point of view, bounds on the convergence rate are crucial for, e.g., tuning controls and analyzing performance. To address this challenge, continuous primal-dual dynamics associated with an augmented Lagrangian have been introduced that are exponentially stable and admit rigorous bounds on the convergence rate~\cite{QL2019, TQL2020}.

The main contribution of this work is to leverage the projection-free primal-dual dynamics~\cite{TQL2020} to develop a novel projection-free power-limiting droop control, establish semi-global exponential stability of the resulting networked multi-converter frequency dynamics with respect to KKT points of an associated constrained flow problem, obtain a bound on the convergence rate of the networked multi-converter dynamics, and analyze the impact of control gains and network parameters on the convergence rate. 

In particular, our projection-free power-limiting droop control results in projection-free networked dynamics whose Carath\'eodory solutions are semi-globally exponentially stable with respect to KKT points of the constrained flow problem introduced in~\cite{IG2025}. To obtain this result, we show that the projection-free networked dynamics corresponds to the primal-dual dynamics of the constrained network flow problem in edge coordinates. Therefore, all existing results on the properties of the KKT points of the constrained flow problem (e.g., synchronous frequency) and steady states of projection-based power-limiting droop control~\cite{IG2025} immediately hold for the proposed projection-free power-limiting droop control.

A key contribution of this work is a bound on the convergence rate of the projection-free networked dynamics. To this end, we characterize the active constraint set and graph of nodes with active constraints. Notably, to evaluate the convergence rate of the resulting projection-free power-limiting droop control, we link the Jacobian matrix of the constraints of the network flow problem to the Laplacian matrix of the graph of nodes with active constraints. This result enables bounding the convergence rate of the networked dynamics as a function of the control gains and properties of the network (e.g., connectivity, maximum node degree, edge weights). In turn, this allows us to propose a control tuning  that improve the bound on the convergence rate. In addition, we show that, under mild technical assumptions, the bound on the convergence rate can be improved by adding edges (e.g., transmission lines) to the graph that increase the connectivity of the graph. Finally, an Electromagnetic transient (EMT) simulation of the IEEE 9-bus system is used to illustrate the results and validate that the proposed control tuning improves control performance.

This paper is organized as follows. Sec.~\ref{sec:preliminary} introduces the power network and converter model, control objectives, and review of projection-based network dynamics. Next, Sec.~\ref{sec:projectionfree} defines a novel projection-free networked dynamics and summarizes the main analytical results. The stability analysis of the networked dynamics in nodal and edge coordinates is presented in Sec.~\ref{sec:stability_analysis}. The sensitivity of the convergence rate with respect to control gains and network parameters is analyzed in Sec~\ref{sec:convrate}. A numerical case study to validate the main results is provided in Sec.~\ref{sec:numerical}. Finally, Sec.~\ref{sec:conclusion} provides conclusions and topics for future work.

\subsection*{Notation}
We use $\mathbb{R}$ and $\mathbb N$ to denote the set of real and natural numbers and define, e.g., $\mathbb{R}_{\geq 0}\coloneqq \{x \in \mathbb R \vert x \geq 0\}$. Moreover, we use $\mathbb{S}_{\succ 0}^n$ and $\mathbb{S}_{\succeq 0}^n$ to denote the set of real positive definite and positive semidefinite matrices. For column vectors $x\in\mathbb{R}^n$ and $y\in\mathbb{R}^m$ we define $(x,y) = [x^\mathsf{T}, y^\mathsf{T}]^\mathsf{T} \in \mathbb{R}^{n+m}$. Moreover, $\norm{x}_{Q}=\sqrt{x^\mathsf{T}Qx}$ denotes the weighted Euclidean norm and $\norm{x}_{\mc C} \coloneqq \min_{z\in\mc C} \norm{z-x}$ denotes the point to set distance. Furthermore, $I_n$, $\mathbbl{0}_{n\times m}$, $\mathbbl{0}_{n}$, and $\mathbbl{1}_n$ denote the $n$-dimensional identity matrix, $n \times m$ zero matrix, and column vectors of zeros and ones of length $n$ respectively. $|\mc X|$ denotes the cardinality of a discrete set $\mc X$. The Kronecker product is denoted by $\otimes$. We use $\varphi_x(t,x_0)$ to denote a (Caratheodory) solution of $\ddt x = f(x)$ at time $t \in \mathbb{R}_{\geq 0}$ starting from $x_0$ at time $t=0$.

\section{preliminaries and network model}\label{sec:preliminary}
    We first introduce the ac power system model, converter model, and control objectives studied throughout the paper.

    \subsection{Power network and converter model}
The power network topology is modeled by a simple, connected and undirected graph $\mathcal{G}\coloneqq \{\mathcal{N}, \mathcal{E}, \mathcal{W}\}$ with edge set $\mathcal{E}\coloneqq\mathcal{N}\times\mathcal{N}$ corresponding to $|\mathcal{E}|=e$ transmission lines, set of nodes $\mathcal{N}$ corresponding to $|\mathcal{N}|=n$ voltage source converters, and set of edge weights $\mathcal{W}=\left\{w_1, \ldots, w_e\right\}$ with $w_i \in \mathbb{R}_{>0}$ for all $i \in \{1,\ldots,e\}$ modeling  line susceptances~\cite{DB2013}. 

    We model each grid-forming voltage source converter $i \in \mathcal{N}$ as a voltage source imposing an ac voltage with phase angle $\theta_i \in \mathbb{R}$ and nominal voltage magnitude that injects an active power denoted by $P_i \in \mathbb{R}$. During normal operation (i.e., outside rare electrical faults), the constraints of the converter and its resource (e.g., batteries, renewables) can be mapped to lower and upper converter power limits are denoted by $P_{\ell} \coloneqq \left(P_{\ell, 1}, \ldots, P_{\ell, n}\right) \in \mathbb{R}^n$ and $P_{u} \coloneqq \left(P_{u, 1}, \ldots P_{u, n}\right) \in \mathbb{R}^n$. Finally, for every $i \in \mathcal{N}$, we use $P_{L, i} \in \mathbb{R}$ to express load mapped to converter buses using Kron-reduction~\cite{DB2013}.

    In the context of frequency stability analysis of transmission systems, line losses are typically negligible. Thus, for brevity of the presentation, we use the widely accepted lossless network model with decoupled frequency and voltage dynamics. The results can be extended to lossy networks by leveraging rotated power measurements~\cite[Sec.~IV-C]{RLB+2012}. Linearizing the ac power flow equation at the nominal voltage magnitude and zero angle difference between nodes, results in the vector $P = \left(P_1, \ldots, P_n\right) \in \mathbb{R}^n$ of converter power injections
    \begin{align}\label{eq:dcpfeq}
        P \coloneqq L\theta + P_{L}, 
    \end{align}
    where $L\coloneqq BWB^\mathsf{T}$ is the Laplacian matrix of the graph $\mc G$, $B \in \{-1,0,1\}^{n \times e}$ denotes the oriented incidence matrix of $\mc G$, and $W=\diag\{w_i\}_{i=1}^{e}$. Moreover, $\theta = \left(\theta_1, \ldots, \theta_n\right) \in \mathbb{R}^n$ is the vector of ac voltage phase angles (relative to $\omega_0 t$ with nominal frequency $\omega_0 \in \mathbb{R}_{> 0}$) and $P_{L} \coloneqq \left(P_{L,1}, \ldots, P_{L, n}\right) \in \mathbb{R}^n$ is the vector of active power loads at every node. The linearized model \eqref{eq:dcpfeq} is widely accepted for studying frequency stability. Broadly speaking, in the context of typical converter-dominated transmission systems, the nonlinearity of the power flow is negligible within the feasible set of converter power injections (see Sec.~\ref{sec:numerical}).
    
    The following assumption formalizes that the converters need to have sufficient capacity to serve the load. In practice this is ensured during system planning.
\begin{assumption}[\textbf{Feasible injection limits and loads}]\label{assum:feas}
For all $i \in \mathcal{N}$, the limits $P_{\ell,i} \in \mathbb{R}^n$ and $P_{u,i} \in \mathbb{R}^n$ satisfy $P_{\ell,i} < P_{u,i}$. Moreover, the disturbance input $P_L \in \mathbb{R}^n$ satisfies $\sum_{i=1}^n P_{\ell,i} < \sum_{i=1}^n  P_{L,i} < \sum_{i=1}^n P_{u,i}$.

\end{assumption}

\subsection{Control Objectives}
Our objective is to design a \emph{decentralized} feedback controller that uses only information available at each node $i \in \mc N$ (e.g., power injection, frequency) to control the power system to the solution of constrained flow problem (CFP) given by
        \begin{align}\label{eq:deCFP}
            \min_{\theta, P} \quad \tfrac{1}{2} \|P-P^\star\|^2_{M} \quad \text {s.t.} \quad  P_\ell \leq L \theta + P_{L}  \leq P_{u} 
            \end{align}
where $M \coloneqq \diag \{m_i\}_{i=1}^{n} \in \mathbb{S}_{>0}^{n}$ is a diagonal matrix of droop gains, $P^{\star}\coloneqq \left(P^\star_1, \ldots, P^\star_{n}\right) \in \mathbb{R}^n$  is a vector of power references periodically
prescribed by the system operator. Including additional constraints (e.g., line limits) is seen as interesting topic for future work. The following assumption formalizes that the converters are dispatched within their power limits. In practice this is ensured  by the system operator through, e.g., security constraint economic dispatch.
\begin{assumption}[\textbf{Feasible references}]\label{assum:setpoint}
    The setpoints $P^\star_i \in \mathbb{R}^n$ satisfy $P_{\ell,i} < P^\star_i < P_{u,i}$.
\end{assumption}

\begin{remark}[\textbf{Distributed algorithms}]\label{rem:distributed}
    Using common algorithms (e.g., primal-dual dynamics\cite{CMC2016,LZC2016}) to solve the CFP \eqref{eq:deCFP} results in \emph{distributed} dynamics that require information exchange between nodes. For example, using primal-dual dynamics~\cite{CMC2016,LZC2016}, the dual multiplier dynamics only depend on local measurements (i.e., power injection) but the primal dynamics (i.e., voltage phase angle) require dual multiplers from neighboring nodes~\cite[Sec.~II-C]{IG2025}. Instead, we seek a \emph{decentralized grid-forming}  controller that maps the converter power injection $P_i$ to its ac voltage frequency $\omega_i = \ddt \theta_i$ and, in contrast to~\cite{IG2025}, admits rigorous bounds on the convergence rate to inform controller tuning and clarify the impact of the network topology.
\end{remark}
Next, we show that feasibility of the CFP \eqref{eq:deCFP} is ensured by Assumption~\ref{assum:feas}.
\begin{proposition}[\textbf{Feasibility in nodal coordinates \cite{IG2025}}]\label{prop:feas}
    There exists $\theta \in \mathbb{R}^n$ such that $P_{\ell} < L\theta + P_{L} < P_{u}$ if and only if $P_\ell$, $P_u$, and $P_L$ satisfy Assumption~\ref{assum:feas}.
\end{proposition}
    Finally, we define $\mathcal{V}_{\theta}$ as the set of KKT points of \eqref{eq:deCFP} (for details see \cite[Sec.~II-B]{IG2025}). 
\begin{definition}[\textbf{KKT points in nodal coordinates}]\label{def:stheta}
$\mathcal{V}_{\theta} \subseteq \mathbb{R}^{3n}$ denotes the points $(\theta^\star,\lambda^\star_\ell,\lambda^\star_u)$ that satisfy the KKT conditions of \eqref{eq:deCFP}, i.e., $P_\ell \leq L \theta^\star + P_L \leq P_u$, $(\lambda^\star_\ell,\lambda^\star_u) \in \mathbb{R}^{2n}_{\geq 0}$, and 
\begin{subequations}\label{eq:KKT:nodal}
\begin{align}
    M (L \theta^\star + P_L-P^\star) + (\lambda^\star_u-\lambda^\star_\ell) &\in \ker{B^\mathsf{T}},\label{eq:KKT:nodal:optimality}\\
    \diag\{\lambda^\star_{\ell,i}\}_{i=1}^n (P_\ell - L \theta^\star - P_L)&=\mathbbl{0}_n,\\
    \diag\{\lambda^\star_{u,i}\}_{i=1}^n (L \theta^\star + P_L - P_u)&=\mathbbl{0}_n.
\end{align}
\end{subequations}
\end{definition}

\subsection{Review of projection-based network dynamics}
Computing the optimal solution of \eqref{eq:deCFP} via its associated primal-dual dynamics~\cite{CMC2016} results in a distributed algorithm that requires exchanging dual-multipliers between nodes. This is not feasible on primary frequency control timescales in large-scale power systems. In contrast, the projection-based network dynamics depicted in Fig.~\ref{fig:projeciton-based}, that resemble but are distinct from the well-known primal-dual dynamics, solve \eqref{eq:deCFP} using only local information~\cite{IG2025}. We require the following definition of projection operator to formalize the projection-based and projection-free network dynamics.
%

\begin{definition}[\textbf{Projection}]\label{def:projection}
    Given a convex set $\mathcal{C} \subseteq \mathbb{R}^n$ and a vector $v \in \mathbb{R}^n$, $\Pi_{\mathcal{C}}(v)$ denotes the projection of $v$ with respect to the set $\mathcal{C}$, i.e., $\Pi_{\mathcal{C}}(v) = \argmin\nolimits_{p \in \mathcal{C}} \norm{p - v}$.
\end{definition}
%
\begin{figure}[!t]
    \includegraphics[width=\columnwidth]{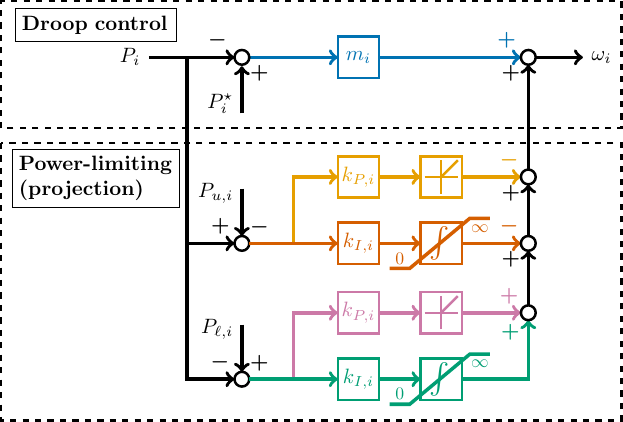}
    \caption{Projection-based power-limiting droop control}
    \label{fig:projeciton-based}
\end{figure}
Next, consider the projection-based power-limiting droop control with states $\ddt \theta_i = \omega_i$,  $\lambda_{u,i} \in \mathbb{R}_{\geq 0}$, and $\lambda_{\ell,i} \in \mathbb{R}_{\geq 0}$ that correspond to the ac voltage phase angles and integral of the upper and lower power limit violations (see Fig.~\ref{fig:projeciton-based}). Let 
\begin{align*}
    g(P_N) = \begin{bmatrix}
    g_\ell(P_N)\\g_u(P_N)
    \end{bmatrix} \coloneqq \begin{bmatrix}
    P_{\ell} - P_N - P_L\\
    P_N + P_{L}-P_u
    \end{bmatrix}.
\end{align*}
The interconnection of the nodal dynamics in Fig.~\ref{fig:projeciton-based} via \eqref{eq:dcpfeq} can be written as the projected dynamical system 
\begin{subequations}\label{eq:plimdroop}
    \begin{align}
        \ddt\theta =& M\left(P^\star-P_L-L \theta\right) - \left(\Xi \otimes K_I \right)\lambda \label{eq:plimdroop:theta}\\ 
            &-\left( \Xi \otimes K_P \right) \Pi_{{\mathbb{R}}^{2n}_{\geq 0} }\left(g(L \theta)\right), \nonumber \\
            \ddt\lambda =&  \Pi_{\mathcal{T}_{\lambda} \mathbb{R}^{2n}_{\geq 0}} \big(\left(I_2 \otimes K_{I}\right) g(L\theta)\big), \label{eq:plimdroop:lambda}
    \end{align}
\end{subequations}
where $\Xi \coloneqq (-1, 1)^\mathsf{T}$ and $\lambda \coloneqq  (\lambda_\ell,\lambda_u) \in \mathbb{R}^{2n}_{\geq 0}$ collects the integrator states. Moreover, the matrices $K_P \coloneqq \diag\{k_{P, i}\}_{i=1}^{n}$ and $K_I \coloneqq \diag\{\sqrt{k_{i}}\}_{i=1}^{n}$ collects proportional $k_{P, i} \in \mathbb{R}_{>0}$ and integral gains $k_{i} \in \mathbb{R}_{>0}$, respectively. We emphasize that this model assumes that the load $P_L$, power setpoints $P^\star$, and power limits $P_\ell$ and $P_u$ are constant on the time-scales of interest for studying frequency stability.

While the multi-converter network dynamics \eqref{eq:plimdroop} do not coincide with primal-dual dynamics of \eqref{eq:deCFP} in nodal coordinates, they coincide after transformation to edge coordinates. Notably, using $\eta = VB^\mathsf{T}\theta$ to transform the multi-converter network dynamics \eqref{eq:plimdroop} results in dynamics that coincide with the primal-dual dynamics of the CFP \eqref{eq:pfangledifference} in edge coordinates (see~\cite[Fig.~1]{IG2025}). Building upon this observation and the LaSalle function from \cite{CMC2016}, it can be shown that \eqref{eq:plimdroop} converges to optimal (i.e., KKT) points of \eqref{eq:deCFP}~\cite{IG2025}.

\section{Projection-free power limiting droop control \& summary of main results}\label{sec:projectionfree}
The discontinuity of the dual dynamics in the projection-based network dynamics \eqref{eq:plimdroop} precludes exponential convergence and significantly complicates bounding the rate of convergence to the KKT points of \eqref{eq:deCFP}.
To overcome the conceptual limitations of projection-based power-limiting droop control, we introduce the novel projection-free power-limiting droop control (see Fig.~\ref{fig:projeciton-free}) that are the main focus of this paper.

\subsection{Projection-free networked dynamics}\label{subsec:projectionfree}
\begin{figure}[!t]
    \includegraphics[width=\columnwidth]{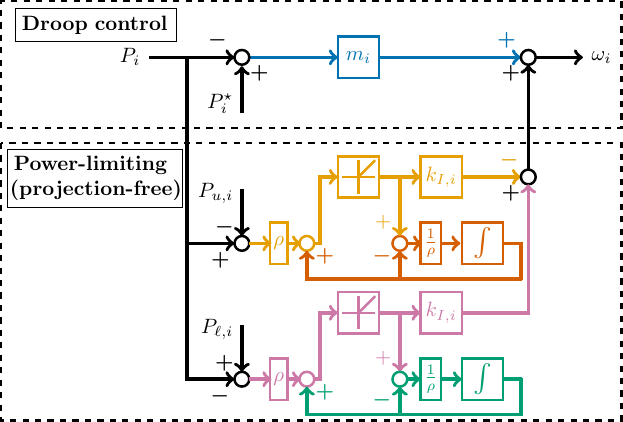}
    \caption{Projection-free power-limiting droop control}
    \label{fig:projeciton-free}
\end{figure}
The control input is the frequency $\omega_i$ of the voltage phase angle $\theta_i$ of each VSC, i.e., $\ddt \theta_i = \omega_i \in \mathbb{R}$. The controller equations are given by
\begin{subequations}\label{eq:newplimdroopsc}
    \begin{align}
    \!\!\!\!\!\!\!\!\!\omega_i =& m_i (P^\star_i-P_i)\! -\! k_i \Pi_{{\mathbb{R}}_{\geq 0}}(\rho(P_i-P_{u,i}) + \lambda_{u, i}) \\
     &+ k_i \Pi_{{\mathbb{R}}_{\geq 0} }(\rho(P_{\ell,i} - P_i) + \lambda_{\ell, i}) , \nonumber\\
    \!\!\!\!\!\!\!\!\!\ddt\lambda_{\ell,i} =&  \frac{1}{\rho}\big(\Pi_{\mathbb{R}_{\geq 0}} \left(\rho(P_{\ell,i} - P_i) + \lambda_{\ell, i}\right) - \lambda_{\ell, i}\big),\\
    \!\!\!\!\!\!\!\!\!\ddt\lambda_{u,i} =&  \frac{1}{\rho}\big(\Pi_{\mathbb{R}_{\geq 0}} \left(\rho(P_i - P_{u,i}) + \lambda_{u, i}\right) - \lambda_{u, i}\big),
    \end{align}
\end{subequations}
where $\lambda \coloneqq (\lambda_{\ell, i}, \lambda_{u, i}) \in \mathbb{R}^{2n}_{\geq 0}$ are the integrals of the violation of the lower and upper power limits, respectively. Moreover, $k_{i}$ is a controller gain and $\rho$ is the integrator time constant. The droop gain of each VSC is denoted by $m_{i}$ \cite{CDA1993}.
This controller is motivated by the projection-free primal-dual dynamics introduced in~\cite{QL2019,TQL2020}. The reminder of this manuscript analyzes the interconnection of the nodal dynamics \eqref{eq:newplimdroopsc} via \eqref{eq:dcpfeq} given by
\begin{subequations}\label{eq:vector_net_nodal_orig}
    \begin{align}
        \!\!\!\ddt \theta &\!=\! M(P^\star \!-\! L \theta \!-\! P_L) \!-\! (\Xi \otimes K_I) \Pi_{\mathbb{R}^{2n}_{\geq 0}}\!\left(\rho g(L\theta) \!+\! \lambda\right)\!,\!\\ 
        \!\!\rho \ddt \lambda &= \Pi_{\mathbb{R}^{2n}_{\geq 0}}(\rho g(L\theta)  + \lambda) - \lambda.
    \end{align}
\end{subequations}
We emphasize that the projection-free multi-converter frequency dynamics \eqref{eq:vector_net_nodal_orig} do not coincide with the projection-free primal-dual dynamics associated with the CFP \eqref{eq:deCFP} in nodal coordinates. However, applying the change of coordinates to edge coordinates, the projection-free multi-converter frequency dynamics coincide with the projection-free primal-dual dynamics of \eqref{eq:pfangledifference}. Thus, \cite[Theorem 1]{TQL2020} can be used to establish exponentially convergence of the  multi-converter frequency dynamics to KKT points of \eqref{eq:deCFP} and bound the convergence rate.

Before establishing stability of the network dynamics, we need to extend the definition of semi-global exponential stability \cite[Theorem 5.17]{S1999} to semi-global exponential stability with respect to a set.

\begin{definition} [\textbf{Semi-global exponential stability}]\label{def:semiglobal}
    Consider the dynamical system $\ddt z = f(z)$. The system is semi-globally
    exponentially stable with respect to  $\mathcal{Z}_e$, if for any $h \in \mathbb{R}_{>0}$, there exist $\mathcal{M}_{\beta} \in \mathbb{R}_{>0}$ and $\beta \in \mathbb{R}_{>0}$ such that for any initial point $z_{0}$ such that $\norm{z_{0}}_{\mathcal{Z}_e} \leq h$, the corresponding solution $\varphi_{z}(t, z_0)$ of the dynamical system satisfies
    \begin{align*}
        \norm{\varphi_z(t, z_0)}_{\mathcal{Z}_e} \leq \mathcal{M}_{\beta} \cdot e^{-\beta t} \norm{\varphi_z(0, z_0)}_{\mathcal{Z}_e}, \;  \forall t\in[0, \infty].
    \end{align*}
\end{definition}

Although Definition~\ref{def:semiglobal} does not place restrictions on $\mathcal{Z}_{e}$, in our analysis $\mathcal{Z}_{e}$ is a convex set of KKT conditions of, e.g., the CFP \eqref{eq:deCFP}.

\begin{theorem}[\textbf{Semi-global exponential stability of networked dynamics in nodal coordinates}]\label{thm:nodal_exp_stability}
    Consider $P_\ell$, $P_u$, $P_L$, and $P^\star$ such that Assumption~\ref{assum:feas} and Assumption~\ref{assum:setpoint} hold. For any connected graph $\mc G$, \eqref{eq:vector_net_nodal_orig} is semi-globally exponentially stable on $\mathbb{R}^n \times \mathbb{R}^{2n}_{\geq 0}$ with respect to the set  $\mathcal{V}_{\theta}$. Moreover, there exists $\omega_s \in \mathbb{R}$ such that $\lim_{t \rightarrow \infty} \omega_i(t) = \omega_{s}$ for all $i \in \mc N$.
\end{theorem}
A proof is provided in the Appendix.

\subsection{Sensitivity of convergence rate bound}
The bound $\beta \in \mathbb{R}_{>0}$ on the convergence rate of the networked dynamics depends on the control gains and network parameters (see Sec.~\ref{sec:convrate} for detailed results). Thus, $\beta \in \mathbb{R}_{>0}$ can be used to inform controller tuning and analyze the impact of the network on control performance. To this end, the next proposition introduces a method to scale any given controller gains $\rho$ and $k_{i}$ to improve the convergence rate.

\begin{theorem}[\textbf{Sensitivity of the convergence rate to control gains}]\label{thm:controlgainconvergence}
    Consider controller gains $k_i \in \mathbb{R}_{> 0}$ and $\rho \in \mathbb{R}_{> 0}$. Moreover, let $k_{i}^\prime = s k_i$ and $\rho^\prime = \frac{1}{\sqrt{s}}\rho$. For all $s \in \mathbb{R}_{> 1}$, it holds that 
     $   \beta(k_i, \rho) < \beta(k^\prime_{i}, \rho^\prime)$.
    \end{theorem}    
A proof is provided in the Appendix. This result establishes that the convergence rate increases for increasing $s \in \mathbb{R}_{>1}$, i.e., increasing $k_i$ and decreasing $\rho$ according to Theorem~\ref{thm:nodal_exp_stability}.

Next, let $d_i$, $d_{\max} \coloneqq \max_{i \in \mathcal{N}} d_i$ and $w_{\max}$ denote the degree of node $i \in \mathcal{N}$, maximum node degree, and maximum edge weight, respectively. Under mild technical assumptions, it can be shown that adding an edge (i.e., transmission line) that increases the degree of nodes that do not have maximum node degree (i.e., $d_l < d_{\max}$) does not decrease the bound $\beta$ on the convergence rate. 

\begin{proposition}[\textbf{Sensitivity of the convergence rate to connectivity}]\label{prop:connectivity:simplified}
Consider a graph $\mc G$ and the corresponding convergence rate $\beta_{\mc G}$. Moreover, consider a graph $\mc G^\prime$ obtained by adding an edge with weight $w_{e+1} \leq w_{\max}$ between any two nodes $(i,j) \notin \mathcal{E}$ of $\mc G$ that satisfy $d_l < d_{\max}$, $l \in {i,j}$, and let $\beta_{\mc G^\prime}$  denote the corresponding convergence rate. Then, there exists $\rho \in \mathbb{R}_{>0}$ such that $\beta_{\mc G^\prime} \geq \beta_{\mc G}$.
\end{proposition}
The result directly follows from Proposition~\ref{prop:connectivity} (see Sec.~\ref{subsec:networktop}) and, broadly speaking, is in line with previous results in grid-forming control that suggest that increasing the connecitivity between the most strongly coupled nodes can be deterimental (i.e., result in instability).

\section{stability analysis}\label{sec:stability_analysis}

As discussed in Remark~\ref{rem:distributed}, applying primal-dual dynamics to solve the CFP~\eqref{eq:deCFP} in nodal coordinates, results in dynamics that cannot be implemented using decentralized converter controls. However, transforming the CFP to edge coordinates~\cite{ZM2011} and applying primal-dual dynamics in edge coordinates results in dynamics that coincide with the projection-free networked dynamics \eqref{eq:newplimdroopsc} that arise from using decentralized projection-free power-limiting droop control (see Fig.~\ref{fig:projeciton-free}).

\subsection{Network and converter model in edge coordinates}
Using the oriented incidence matrix $B$ and $V \coloneq W^\frac{1}{2} \in \mathbb{R}^{n \times n}$, we define the change of coordinates $\eta = VB^\mathsf{T}\theta$ to  so-called edge coordinates~\cite{ZM2011}. The power injection \eqref{eq:dcpfeq} in edge coordinates follows $P \coloneq BV\eta + P_{L}$. Next, consider the constrained flow problem in edge coordinates 
\begin{subequations}\label{eq:pfangledifference}
\begin{align}
    &\min_\eta \quad \tfrac{1}{2} \norm{ BV\eta}^2_{M}   +\left(P_L-P^\star\right)^{\mathsf{T}} M BV\eta 
    \label{eq:pfangledifference:objective} \\ 
    & \text{s.t.} \quad K_I P_{\ell} \leq K_I (BV\eta+P_L)  \leq K_I P_u. 
    \label{eq:pfangledifference:constraints}
\end{align} 
\end{subequations}
\begin{assumption}\label{assum:LICQ}
    The linear independence constraint qualification (LICQ) holds at any optimizer $\eta^\star$ of \eqref{eq:pfangledifference}.
\end{assumption}

We highlight that Assumption~\ref{assum:LICQ} is only used to establish stability and convergence properties of the primal-dual dynamics associated with \eqref{eq:pfangledifference} but is not needed to relate the primal-dual dynamics associated with \eqref{eq:pfangledifference} to the projection-free network dynamics presented in Sec.~\ref{subsec:projectionfree}.

Finally, we define $\mathcal{V}_{\eta}$ as the set of KKT points of \eqref{eq:pfangledifference} (for details see \cite[Sec.~III-B]{IG2025}).
\begin{definition}[\textbf{KKT points of CFP in edge coordinates}]\label{def:seta}
$\mathcal{V}_{\eta} \subseteq \mathbb{R}^{e + 2n}$ denotes the set of points $(\eta^\star,\lambda^\star_\ell,\lambda^\star_u)$ that satisfy the KKT conditions of the CFP in edge coordinates \eqref{eq:pfangledifference}, i.e., $P_\ell \leq BV \eta^\star + P_L \leq P_u$,  $(\mu^\star_\ell,\mu^\star_u) \in \mathbb{R}^{2n}_{\geq 0}$, and 
\begin{subequations}\label{eq:KKT:edge} 
\begin{align} 
    M (BV \eta^\star + P_L - P^\star) + K_{I}(\mu^\star_u-\mu^\star_\ell) &\in \ker(B^\mathsf{T}),\label{eq:KKT:edge:optimality}\\
    \diag\{\mu^\star_{u,i}\}_{i=1}^n K_{I} (BV \eta^\star + P_L - P_u)&=\mathbbl{0}_n,\\
    \diag\{\mu^\star_{\ell,i}\}_{i=1}^n K_{I} (P_\ell - BV \eta^\star - P_L)&=\mathbbl{0}_n.
\end{align}
\end{subequations}
\end{definition}
\subsection{Stability in edge coordinates}
The Augmented Lagrangian introduced in \cite{B1982} associated with the CFP in edge coordinates \eqref{eq:pfangledifference} is given by 
\begin{align}\label{eq:auglagrangian}
    \mathcal{L}(\eta, \mu_{\ell}, \mu_{u}) = \tfrac{1}{2} \norm{ BV\eta}^2_{M}   +\left(P_L-P^\star\right)^{\mathsf{T}} M BV\eta \nonumber\\
    + \sum_{j=1}^{m} \mathcal{H}_{\rho}(-a_j^\mathsf{T} \eta - b_j, \mu_{j, \ell}) +  \mathcal{H}_{\rho}(a_j^\mathsf{T} \eta - c_j, \mu_{j, u})
\end{align}
and the penalty function
\begin{subequations}
    \begin{align*}
        &\mathcal{H}_{\rho}(a_j^\mathsf{T}\eta - b_j, \mu_j) \\
        &= \begin{cases*}
            \makebox[12.5em][l]{$(a_j^\mathsf{T}\eta - b_j)\mu_j + \frac{\rho}{2}(a_j^\mathsf{T}\eta - b_j)^2$}
            \text{if } \rho(a_j^\mathsf{T}\eta - b_j) + \mu_j \geq 0 \\
            \makebox[12.5em][l]{$-\frac{1}{2} \frac{\mu_j^2}{\rho}$}
            \text{if } \rho(a_j^\mathsf{T}\eta - b_j) + \mu_j < 0
        \end{cases*}.
    \end{align*}
\end{subequations}
Here, $a_j$ is the $j$-th column of $(K_{I}BV)^\mathsf{T}$, $b_j$ is the $j$-th element of $K_I(P_L - P_{\ell})$, $c_j$ is the $j$-th element of $K_I(P_u - P_L)$, and  $\mu_{\ell} \in \mathbb{R}_{\geq 0}$ and $\mu_{u} \in \mathbb{R}_{\geq 0}$ are the dual multipliers associated with the lower and upper inequalities respectively.  

Note that $\sum_{j=1}^{n} x_j \Pi_{\mathbb{R}_{\geq 0}}(y_j) = X^\mathsf{T} \Pi_{\mathbb{R}^n_{\geq 0}}(y)$, where $x_j$ is the $j$-th column of matrix $X$ and $y_j$ is the $j$-th element of the column vector $y$. Then, primal-dual dynamics associated with \eqref{eq:auglagrangian} are given by $\mu = (\mu_\ell,\mu_u)$ and
\begin{subequations}\label{eq:primaldualedge}
    \begin{align}
        \ddt \eta =& VB^\mathsf{T}\Big(\! M(P^\star - P_{L} - BV\eta) \label{eq:ddteta} \\
        &- (\Xi \otimes K_I) \Pi_{\mathbb{R}^{2n}_{\geq 0}}\big(\rho (I_2 \otimes K_I) g(BV\eta) \!+ \!\mu\big)\Big) \nonumber \\
        \rho  \ddt \mu =& \Pi_{\mathbb{R}^{2n}_{\geq 0}}\big(\rho (I_2 \otimes K_I) g(BV\eta) + \mu\big) - \mu.
    \end{align}
\end{subequations}


The next theorem establishes the projection-free networked dynamics \eqref{eq:vector_net_nodal_orig} and projection-free primal-dual dynamics \eqref{eq:primaldualedge} coincide in edge coordinates. To this end, let $T_\eta \coloneqq \blkdiag(VB^\mathsf{T}, K_I)$.
\begin{proposition}[\textbf{Coinciding dynamics}]\label{prop:equivalency}
    Consider any initial condition $(\theta_0,\lambda_0)$ of the projection-free networked dynamics \eqref{eq:vector_net_nodal_orig} in nodal coordinates. Moreover, consider the corresponding initial condition $(\eta_0,\mu_0)=T_\eta (\theta_0, \lambda_0)$ of the projection-free primal-dual dynamics \eqref{eq:primaldualedge} of \eqref{eq:auglagrangian} in edge coordinates. Then, the solutions $\varphi_{\theta}(t,(\theta_0,\lambda_0))$ and $\varphi_{\eta}(t,(\eta_0, \mu_0))$ of \eqref{eq:vector_net_nodal_orig} and \eqref{eq:primaldualedge} satisfy $T_\eta \varphi_{\theta}(t,(\theta_0,\lambda_0)) = \varphi_{\eta}(t,T_\eta(\theta_0, \lambda_0))$.
\end{proposition}
A proof is provided in the Appendix.

\begin{proposition}[\textbf{Semi-global exponential stability of primal-dual dynamics in edge coordinates}]\label{prop:nodal_exp_stability}
    Consider $P_\ell$, $P_u$, $P_L$, and $P^\star$ such that Assumption~\ref{assum:feas} and Assumption~\ref{assum:setpoint} hold. Then, under Assumption~\ref{assum:LICQ}, the primal-dual dynamics \eqref{eq:primaldualedge} are semi-globally exponentially stable with respect to $\mathcal{V}_{\eta}$ on $\mathbb{R}^{e} \times \mathbb{R}^{2n}_{\geq0}$. Moreover, $\ddt (\eta,\lambda) = \mathbbl{0}_{e+2n}$ holds on $\mathcal{V}_{\eta}$.
\end{proposition} 
A proof is provided in the Appendix. We emphasize that the frequency and active constraints upon convergence for both projection-based and projection-free power limiting droop control are identical and solely determined by the KKT points of \eqref{eq:deCFP}. In particular, synchronous frequency $\omega_s$ can be explicitly expressed as a function of the active sets $\mc I_u$ and $\mc I_\ell$, total load $\sum_{i \in \mc N}  P_{L,i}$, total power dispatch $\sum_{i \in \mc N} P^\star_i$, and droop coefficients~\cite[Th.~3]{IG2025}.

We emphasize that global asymptotic stability of the projection-free networked dynamics \eqref{eq:newplimdroopsc} in nodal coordinates can be established by applying arguments from~\cite{TQL2020,IG2025} to the primal-dual dynamics associated with the augmented Lagrangian~\eqref{eq:auglagrangian}. However, since the gradient $K_I BV$ of the constraints~\eqref{eq:pfangledifference:constraints} with respect to $\eta$ is not full row rank, global exponential stability can generally not be achieved~\cite{TQL2020,DH2019}. Thus, we focus on establishing semi-global exponential stability of the projection-free networked dynamics. We  begin by introducing several constants needed for our analysis.
\subsection{Graph of nodes with active constraints}
Using the following definitions, we introduce the graph of converters with active. Broadly speaking, inverters operating at the limit are not providing frequency control and hence stability and convergence crucially hinges on the number and location of converters that are limited. Intuitively, the frequency dynamics cannot be stable cannot be stable if all converters are operating at their limit.
\begin{definition}[\textbf{Active constraint sets}] \label{def:active_set}
    We define $\mathcal{I}_{\ell} \subseteq \mc N$ and $\mathcal{I}_{u} \subseteq \mc N \setminus \mathcal{I}_{\ell}$ as the set of nodes at their lower and upper limit, i.e., $i \in \mathcal{I}_{\ell}$ if and only if $P_i = P_{\ell, i}$ and $i \in \mathcal{I}_{u}$ if and only if $P_i = P_{u, i}$. Moreover, we define $\mathcal{I} \coloneqq \mathcal{I}_{\ell} \cup \mathcal{I}_u$.
\end{definition}
    \begin{definition}[\textbf{Graph of nodes with active constraint}]\label{def:active_incidence}
        Let $B_{\mathcal{I}}$ contain the rows of $B$ associated with the active set $\mathcal{I}$. Moreover, we define the graph $\mathcal{G}_{\mathcal{I}}$ containing nodes $i \in \mathcal{I}$ and edges $(i,j) \in (\mathcal{I} \times \mathcal{I}) \cap \mc E$. Additionally, edges connecting nodes $i \in \mathcal{I}$ and nodes $j  \notin \mathcal{I}$ appear as self-loops in $\mathcal{G}_{\mathcal{I}}$.
    \end{definition}
    By extracting the weights corresponding to active nodes as $W_{\mathcal{I}}$, we  define the Laplacian matrix associated with $\mathcal{G}_{\mathcal{I}}$ as $L_{\mathcal{I}} = B_\mathcal{I} W_{\mathcal{I}} B_{\mathcal{I}}^\mathsf{T},$ where removing the rows of $B$ corresponding to inactive nodes results in $B_{\mathcal{I}}$ with (i) columns of zeros associated with edges of the inactive part of the graph, (ii) columns containing only $1$ or $-1$ associated with edges connecting the inactive part of the graph to that of the active part and (iii) columns containing $1$ and $-1$ associated with edges of the active part of the graph. It should be noted that the edges connecting the inactive to the active part of the graph represent themselves as self-loops in $L_{\mathcal{I}}$ (see Fig.~\ref{fig:graph}).

    \begin{figure}[!t]
    \vspace*{-1em}
    \centering
    \includegraphics[width=0.8\columnwidth]{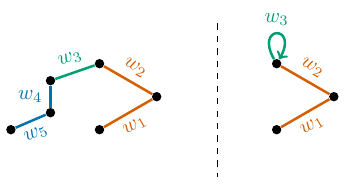}
    \caption{The graph $\mathcal{G}$ and the active graph $\mathcal{G}_{\mathcal{I}}$ corresponding to $L_{\mathcal{I}}$: inactive constraint nodes (blue), active constraint nodes (red), edges connecting active constraint nodes (orange), and edges connecting active and inactive constraint nodes (green).}
    \label{fig:graph}
\end{figure}

Now we can discuss the relation between the Laplacian $L_{\mathcal{I}}$ and the Jacobian of the of CFP constraints in edge coordinates. Let $\mathcal{J}$ denote the Jacobian of $K_I (BV\eta+P_L)$ with respect to $\eta \in \mathbb{R}^e$. At any optimizer $\eta^\star \in \mathbb{R}^e$ of \eqref{eq:pfangledifference}, we define
\begin{align*}
    \kappa \coloneq \lambda_{\min} \left(\mathcal{J}_{\mathcal{I_{\ell}} \cup \mathcal{I}_{u}}\mathcal{J}_{\mathcal{I_{\ell}} \cup \mathcal{I}_{u}}^\mathsf{T}\right),
\end{align*}
where $\mathcal{J}_{\mathcal{I}}$ collects the Jacobians of the rows of $K_I (BV\eta+P_L)$ that correspond to nodes $i \in \mathcal{I}$ with respect to $\eta \in \mathbb{R}^e$.

    \begin{remark}[\textbf{Radial network}]
    For radial networks, $\kappa$ reduces to the minimum eigenvalue of the Laplacian $L_{\mathcal{I}}$ scaled by the controller gains $K_{I_{\mathcal{I}}}$ of the nodes with inactive constraints, i.e, $\kappa = \lambda_{\min}(K_{I_{\mathcal{I}}} L_{\mathcal{I}} K_{I_{\mathcal{I}}})$. Notably, if $n \geq 2$ and at least one converter is not operating at its limits, then $L_{\mathcal{I}}$ is a reduced loopy Laplacian and hence $\lambda_{\min}(K_{I_{\mathcal{I}}}L_{\mathcal{I}}K_{I_{\mathcal{I}}}) >0$. Moreover, it follows from \cite[Prop.~2]{IG2025} that either $\mathcal{I}_{\ell} = \emptyset$ or $\mathcal{I}_{u} = \emptyset$, i.e., all constrained nodes are either at their upper or lower bound. Thus, it holds that either $\kappa = \lambda_{\min} \left(\mathcal{J}_{\mathcal{I_{\ell}}}\mathcal{J}_{\mathcal{I_{\ell}}}^\mathsf{T}\right)$ or $\kappa = \lambda_{\min} \left(\mathcal{J}_{\mathcal{I}_{u}}\mathcal{J}_{\mathcal{I}_{u}}^\mathsf{T}\right)$. Moreover, for uniform controller gains, i.e., $k_i = c \in \mathbb{R}_{> 0}$ for all $i\in \mathcal{N}$, it holds that $\kappa = c \lambda_{\min}(L_{\mathcal{I}})$.
    \end{remark}    

Please see proof of Proposition~\ref{prop:nodal_exp_stability} for similar insights for networks that are not radial. Finally we note that, if all nodes have active constraints (i.e., either reached their lower or upper power limits), it follows that $\kappa=0$ and, As will be shown subsequently, the frequency dynamics are not stable.

\subsection{Preliminary results and bounds}
This section presents several bounds that are required to bound the convergence rate $\beta \in \mathbb{R}_{>0}$. We begin by bounding the maximum eigenvalue of Laplacian $L$.
\begin{lemma}[\textbf{Maximum eigenvalue of the Laplacian}]\label{lem:laplacian_eigenvalues}
Consider the degree $d_i$ of node $i \in \mathcal{N}$, the minimum and maximum edge weights $w_{\min}$ and $w_{\max}$, and maximum node degree $d_{\max} \coloneqq \max_{i \in \mathcal{N}} d_i$. Then, the largest eigenvalue of the Laplacian $L$ is bounded by
\begin{align}\label{eq:laplacian_bound_max}
w_{\min}(1 + d_{\max} ) \leq \lambda_{\max} (L) \leq 2 w_{\max} d_{\max}.
\end{align}
\end{lemma}
A proof is provided in the Appendix. 

\begin{figure*}[t]
    \vspace*{-0.5cm}
  \centering
  \makebox[\textwidth][c]{
    \includegraphics[width=1.2\textwidth]{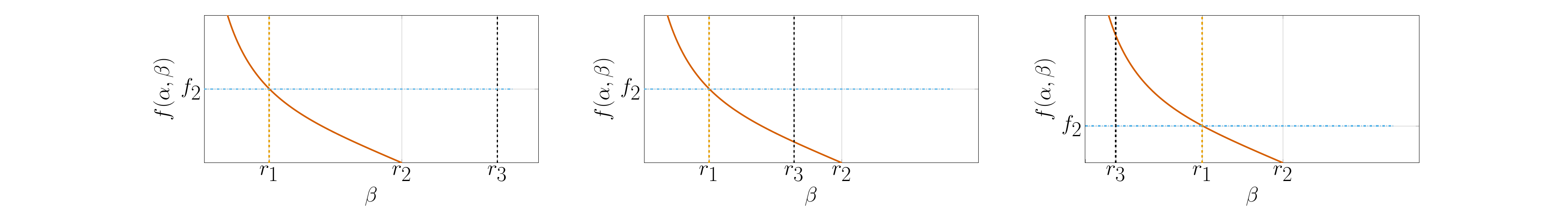}
  }
  \caption{Illustration of the bounds \eqref{eq:beta_constant} and \eqref{eq:beta_cubic} on the convergence rate $\beta$. The functions $f(\alpha, \beta) \coloneqq \frac{\kappa\alpha}{4\beta} - 4\beta^2$  and $f_2 = \mathcal{F}^2 + \frac{\kappa}{4} + (\gamma + \mathcal{M})(\alpha + \mathcal{M} + \frac{1}{\rho}) + \frac{1}{2\rho^2}$ are left and righthand side of the bound \eqref{eq:beta_cubic} respectively, where $r_1 = \arg_{\beta}(f(\alpha,\beta) = f_2)$, $r_2 = \sqrt[3]{\frac{\kappa \alpha}{16}}$, and $r_3 = \frac{\kappa \delta_{\min}}{46\rho \mathcal{F}^2}$.}\label{fig:beta}
\end{figure*}

Let $k_{\max} \coloneqq \max_{i \in \mathcal{N}} k_i$ denote the maximum control gain. We first introduce $\mathcal{S}=\sqrt{2 d_{\max} k_{\max} w_{\max}}$ and $\mathcal{F} = \sqrt{2 d_{\max} k_{\max} w_\Sigma}$ that, loosely speaking, bound the spectral and Frobenius norm. Next, we define
    \begin{align*}
        \delta_{\min}(\rho, k_i) &\coloneq 1 - \bigg[1+\rho \cdot \frac{\max\limits_{i \in N \setminus (\mathcal{I}_u\cup \mathcal{I_\ell})}{\sqrt{k_i} g_i(\eta^\star)}}{\nu_{0}}\bigg]_{+}^2, 
    \end{align*}
and by using $\mathcal{S}$ and $\mathcal{F}$, we have
\begin{align}\label{eq:M_theta}
    \mathcal{M}(\rho, \mathcal{F}, \mathcal{S}) =  \rho \mathcal{F}^2 + \mathcal{S}(\rho \mathcal{F} \nu_{0} + \nu_{0} + \norm{\mu^\star})
\end{align}
where $\nu_{0} \in \mathbb{R}_{>0}$ is the point to set distance of initial primal-dual pair $(\eta_0, \mu_{\ell, 0}, \mu_{u, 0})$ to the set of KKT points $\mathcal{V}_{\eta}$.

Moreover, consider $\alpha \coloneqq  m_{\min} \lambda_{\min}^{+}(L)$, where $\lambda^{+}_{\min}(L)$ denotes the smallest non-zero eigenvalue of $L$, and $\gamma \coloneqq 2w_{\max} m_{\max}d_{\max}$. Loosely speaking, $\alpha$ and $\gamma$ model the weakest and strongest connecitivity between nodes in the network. Broadly speaking, $\alpha$ corresponds to the weakest connection in the graph~\cite{LNS} multiplied by the lowest droop coefficient and $\gamma$ bounds the connecitivity of the node with the strongest connection multiplied by the largest droop coefficient. Then, the convergence rate bound $\beta$~\cite{TQL2020} is bounded by
\begin{subequations}\label{eq:beta_bound}
    \begin{align}
        \beta &\leq \frac{\kappa\delta_{\min}}{46\rho \mathcal{F}^2}, \label{eq:beta_constant}
        \\
        \!\!\frac{\kappa\alpha}{4\beta}\! - \!4\beta^2 &\!\geq\! \mathcal{F}^2 \!+\! \frac{\kappa}{4} \!+\! (\gamma \!+\! \mathcal{M})(\alpha \!+\! \mathcal{M} \!+\! \frac{1}{\rho}) \!+\! \frac{1}{2\rho^2}. \label{eq:beta_cubic}
    \end{align}
\end{subequations}
The bounds on the convergence rate are illustrated in Fig.~\ref{fig:beta}. Finally, we show that the norm $\norm{\mu^\star}$ of the optimal dual multipliers $\mu^\star$ can be bounded independently of $\rho$ and $L$.

\begin{lemma}[\textbf{Bounds on the optimal dual multiplier}]\label{lem:mubound}
For any KKT point of \eqref{eq:deCFP}, it holds that $\norm{\mu^\star} \leq \frac{m_{\max}}{\sqrt{k_{\min}}} \varrho$, where $\varrho =  \sum_{i=1}^n  (P^\star_i - P_{L,i}) - \sum_{i=1}^n  \frac{\omega_s}{m_i}$ and $k_{\min}=\min_{i \in \mc N} k_i$.
\end{lemma}    
A proof is provided in the Appendix.

\section{Sensitivity analysis of the convergence rate}\label{sec:convrate}

\subsection{Bounds on the control gain $\rho$}
Next, we bound the gain $\rho \in \mathbb{R}_{>0}$ to narrow down the range of gains $\rho \in \mathbb{R}_{>0}$ to be considered for controller tuning.

\begin{proposition}[\textbf{Bounded $\rho^\star$}]\label{prop:rho_bound}
    The gain $\rho^\star \in \mathbb{R}_{>0}$ that maximizes the convergence rate bound $\beta \in \mathbb{R}_{>0}$ for the projection-free network dynamics is bounded by $c_1 = \mathcal{F}^2 + \frac{\kappa}{4} + \gamma \alpha$ and
    \begin{align*}
       \frac{2c_1}{23\kappa \alpha} < \rho^\star \leq 1 + \frac{1}{23\kappa\alpha}\max\{2c_1, 2\gamma, \frac{1062}{1058}\}.
    \end{align*}
    \end{proposition}
A proof is provided in the Appendix.

\subsection{ Impact of network topology}\label{subsec:networktop}
In addition to the control gains, the convergence rate of the projection-free networked dynamics \eqref{eq:vector_net_nodal_orig} to the KKT points of \eqref{eq:deCFP} crucially depends on the network topology and edge weights (e.g., transmission line susceptances). In this section, we analyze the impact of key network parameters on the convergence rate.

To this end, recall that $d_i$, $i \in \mathcal{N}$, and $d_{\max}$ denote the node degree and maximum node degree, respectively. The following proposition establishes that, under some mild technical assumptions, the convergence rate $\beta$ increases when the connectivity of the graph $\mc G$ increases.

\begin{proposition}[\textbf{Adding an edge to $\mc G$}]\label{prop:connectivity}
Consider a graph $\mathcal{G}^\prime$ obtained by adding an edge with weight $w_{e+1} \leq w_{\max}$ between any two nodes $(i,j) \notin \mathcal{E}$ of the graph $\mathcal{G}$ that satisfy $d_l < d_{\max}$, $l \in {i,j}$. Let $\beta_{\mc G}$ and $\beta_{\mc G^\prime}$ denote the convergence rate bounds corresponding to the dynamics over the graphs $\mathcal{G}$ and $\mathcal{G}^\prime$. If the set of active nodes does not change and \eqref{eq:beta_constant} is binding, then $\beta_{\mc G} = \beta_{\mc G^\prime}$. Let $k_{\min} > \frac{\nu^2_{0}}{(P_{\text{res}}\rho)^2}$, and assume that $P_{\text{res}} = \max\limits_{i \in N \setminus (\mathcal{I}_u\cup \mathcal{I_\ell})}{-g_i(\eta^\star)}$ does not change after adding the edge. Moreover, assume that $\rho \in \mathbb{R}_{>0}$ satisfies
    \begin{align*}
        \rho(\frac{21}{2} - \sqrt{\frac{w_{\max}}{w_{\Sigma}}}\nu_0) &\leq \frac{m_{\max}w_{\max}}{k_{i, \max} w_{\Sigma}} + \frac{\sqrt{w_{\max}}}{\sqrt{k_{i, \max} d_{\max}}w_{\Sigma}}\nu_{0}\\
            + \frac{m_{\max}}{\sqrt{k_{\min}}\mathcal{F}^2}\varrho.
    \end{align*}
If the set of active nodes does not change and \eqref{eq:beta_cubic} is binding, then $\beta_{\mc G^\prime} \geq \beta_{\mc G}$.
    \end{proposition}
A proof is provided in the Appendix. This results shows that, under mild technical assumptions, adding edges that connect previously weakly connected nodes (i.e., increasing the connecitivity of the network) increases the convergence rate of the projection-free networked dynamics. Specifically, for a sufficiently large gain $\rho \in \mathbb{R}_{>0}$, the convergence rate can be improved by introducing an edge with a weight equal or lower than $w_{\max}$ that connects any two nodes with degree less than $d_{\max}$. This result crucially hinges on Lemma~\ref{lem:mubound}, i.e., that the optimal dual multipliers $\mu^\star$ are independent of the gain $\rho$ and $L$ at any KKT point.

Finally, the following corollary establishes the same result for the simpler setup of a graph with identical weights.       

\begin{corollary}[\textbf{Uniform edge weights}]\label{corr:edge_uniformweight}
    Consider a graph $\mathcal{G}$ with uniform edge weights $w_i = w^\prime$ for all $i \in \{1, \ldots, e\}$, let $k_{i, \max} = s_m m_{\max}$ denote the controller gain as function of the maximum droop gain $m_{\max}$, and let $s_m \in \mathbb{R}_{>1}$. For any connected graph $\mathcal{G}$ with $n\geq 3$ and $e$ edges, the lower bound on $\rho$ in Proposition~\ref{prop:connectivity} can be replaced by
        \begin{align*}
        \left(\tfrac{21}{2}e - \sqrt{e} \nu_{0}\right) \rho \leq \frac{1}{s_m}+\frac{\nu_{0}}{\sqrt{k_{i, \max} d_{\max}w^\prime}} + \frac{e m_{\max}}{\sqrt{k_{\min}}\mathcal{F}^2}\varrho.
    \end{align*}    
\end{corollary}


\begin{figure*}[htbp]
    \centering
    \includegraphics[width=0.95\linewidth]{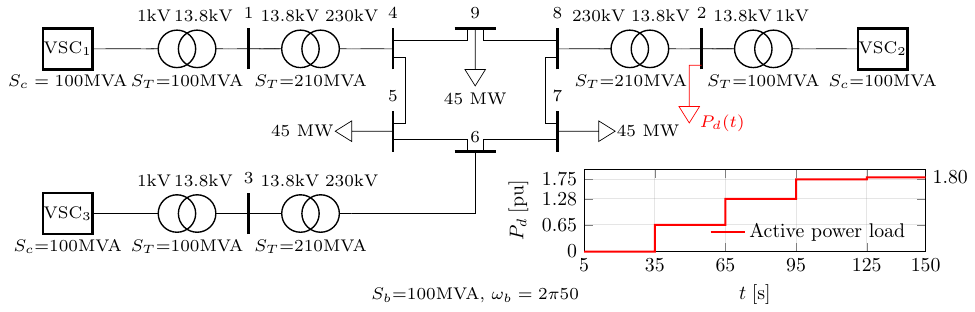}
    \caption{IEEE 9-bus test case system with three two-level voltage source converters and constant impedance (black) and constant power loads (red).        \label{fig:9bus_system}}
\end{figure*}

\section{Numerical Case Study}\label{sec:numerical}
To illustrate and validate our analytical results obtained using the reduced-order model \eqref{eq:newplimdroopsc} and \eqref{eq:dcpfeq}, we use an electromagnetic transient (EMT) simulation of the IEEE 9-bus system (see Fig.~\ref{fig:9bus_system}) with three VSCs controlled by projection-free power limiting droop control.
\subsection{Power system model}
We replaced the three synchronous generators in the IEEE 9-bus system with voltage source converters (VSCs). An average model of two-level VSCs with $LC$ output filter and standard cascaded inner voltage and current loops is used. Details on the transformer parameters, converter parameters, and control gains of the inner control loops can be found in~\cite{TGA+2020}. Table~\ref{table:parameters} and \cite[Table I]{IG2025} summarize the converter rating, power setpoints, power limits, and control gains for power-limiting droop control used in this work. In addition to the standard IEEE 9-bus base load a time-varying constant power load (see Fig.\ref{fig:9bus_system}) is introduced to create overload conditions for the VSCs. Notably, the reduced-order model \eqref{eq:newplimdroopsc} is obtained by applying kron-reduced~\cite{DB2013} to obtain a network model with three buses and assuming that active power and frequency are decoupled from reactive power and voltage magnitude.
\subsection{Simulation results and discussion}
Simulation results are shown in Fig.~\ref{fig:freqresponse} and Fig.~\ref{fig:comparison}. EMT simulation results are shown in the top row and simulation results obtained using the reduced-order model \eqref{eq:newplimdroopsc} are shown in the bottom row. The pink markers in the bottom row of Fig.~\ref{fig:freqresponse} indicate the frequency deviation predicted by \cite[Theorem 3]{IG2025}. Under the same load profile used in \cite{IG2025}, we validate the analytical results for the synchronous frequency, and observe that the analytical results closely match the EMT simulation and results obtained using the reduced-order model \eqref{eq:newplimdroopsc}.

Finally, to compare projection-based~\cite{IG2025} and projection-free power-limiting droop control \eqref{eq:newplimdroopsc} the response of active power to a load increase and subsequent convergence are shown in Fig.~\ref{fig:comparison}. Notably, before the load increase at $t = 95~\mathrm{s}$, VSC 2 is overloaded. In addition, the load increase at $t = 95~\mathrm{s}$ overloads VSC 3. Moreover, after the load increase at $t = 125~\mathrm{s}$ all VSCs operate at their maximum active power. In addition, we validate the improvement of the convergence rate according to Theorem~\ref{thm:nodal_exp_stability} by scaling controller gains, i.e., using the scaling $s=1.66$ improves the convergence rate relative to using the scaling $s=1$.
\input{ieee9bus_table.tex}

\begin{figure*}[htbp]

    \begin{center}

        \includegraphics[width=1\linewidth]{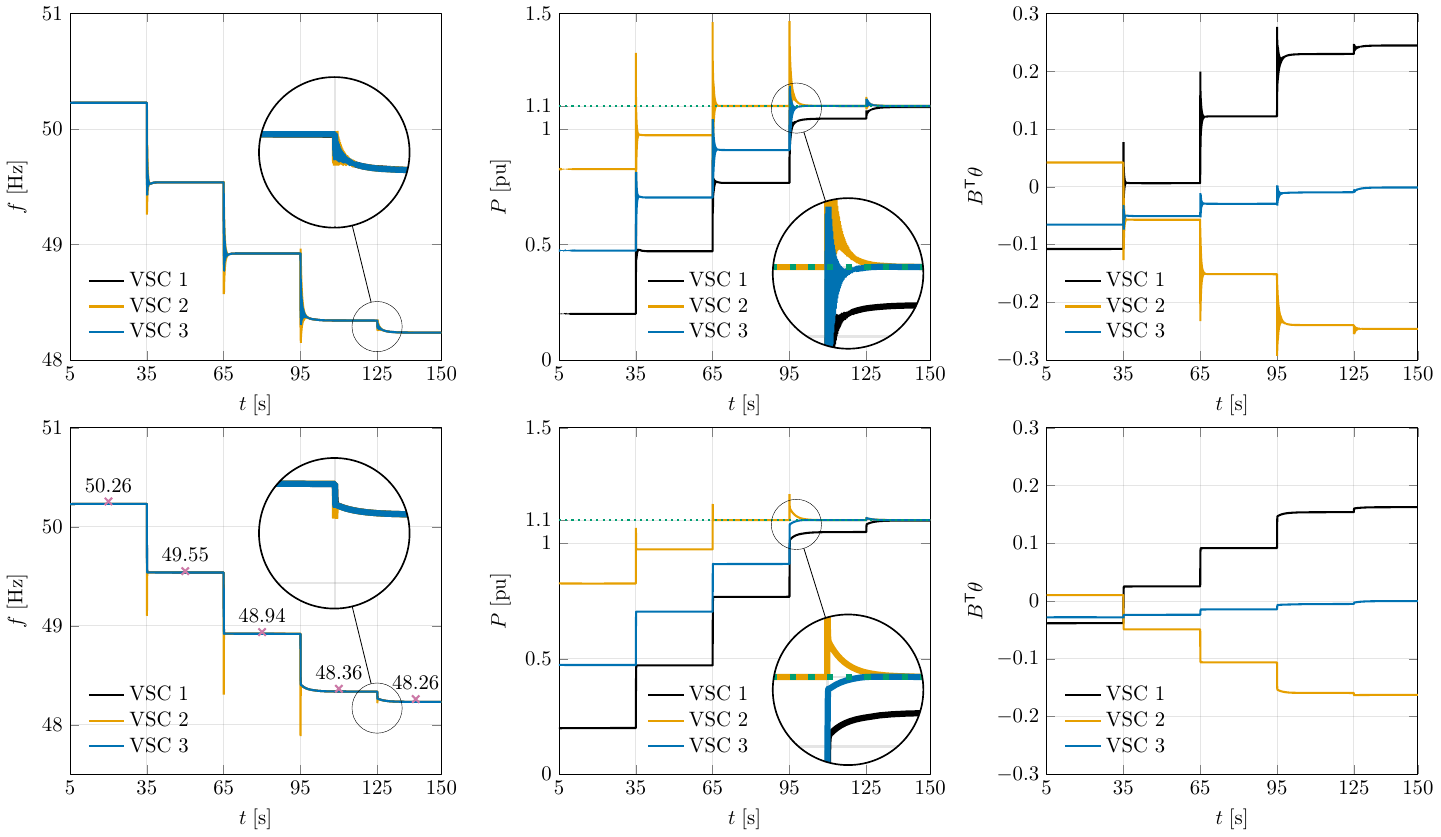}
        \caption{Results of an EMT simulation (top row) and the reduced-order model (bottom row) for the IEEE 9-bus system depicted in \cite[Fig. 3.]{IG2025}. The green line indicates the upper power limit of each VSC and the pink markers indiciate the frequency deviation predicted by \cite[Theorem 3]{IG2025} using the parameters in \cite[Table I]{IG2025}.\label{fig:freqresponse}}
        \hspace*{-2em}
          \makebox[\textwidth][c]{%
        \includegraphics[width=1.23\linewidth]{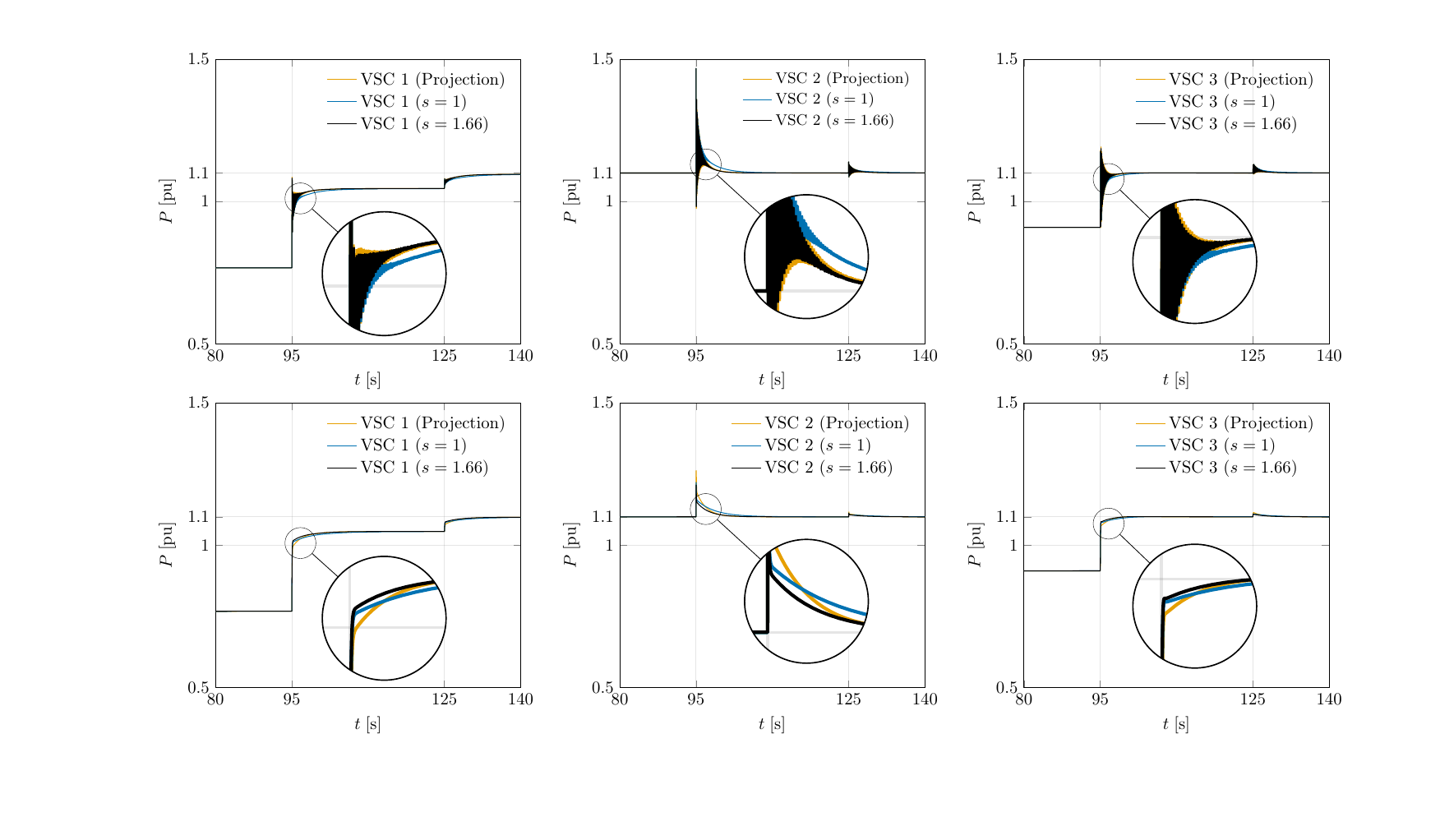}
          }
          \vspace*{-4em}
        \caption{Comparison of an EMT simulation (top row) and the reduced-order model (bottom row) of the IEEE 9-bus using the projection-based and projection-free dynamics for different values of scaling factor $s$.\label{fig:comparison}}
    \end{center}
\end{figure*}

\section{Conclusion and Outlook}\label{sec:conclusion}
In this paper, we studied a constrained network flow problem that aims to minimize the deviation of grid-forming power injections from given references under power limits. Applying primal-dual dynamics to the constrained flow problem in its original coordinates results in dynamics that cannot be implemented using only local measurements. We investigated novel projection-free networked dynamics that leverage local measurements of the power injections to solve the constrained flow problem. We showed that, in edge coordinates, the networked dynamics coincide with projection-free primal-dual dynamics of the constrained flow problem. Then we established that the networked dynamics are semi-globally asymptotically stable with respect to the set of the KKT points of the constrained flow problem in the original nodal coordinates. Leveraging our theoretical results, we (i) provide a bound on the convergence rate of the projection-free networked dynamics, (ii) study the dependence of convergence rate of the networked dynamics on the control gains and introduce a method to scale given controller gains to improve the convergence rate and (iii) analyze the relationship of the bound on the convergence rate and connectivity of the network. Particularly, converters share additional load according to their droop coefficients up to their power limit. Finally, the analytical results are illustrated using an EMT simulation. While these results are encouraging, future work should consider a wider range of constraints typically encountered in the power system context such as converter dc voltage, ac current limits and line flow limits.

\section{Appendix}\label{sec:appendix}
\begin{proof}[Proof of Theorem~\ref{thm:nodal_exp_stability}]
Definition~\ref{def:semiglobal} does not restrict the initial condition and we consider nonnegative initial values $\lambda(0), \mu(0) \geq 0$. In this case, the dual multipliers remain in nonnegative orthant~\cite[Proposition 1]{TQL2020}. Next, we can establish the stability of the dynamical system in original coordinates. Notably, there exist $\underline{\kappa} \in \mathbb{R}_{\geq 0}$ and $\overline{\kappa} \in \mathbb{R}_{\geq 0}, \forall t \geq 0$ such that 
\begin{align}\label{eq:normbound}
    \underline{\kappa} \norm{T_\eta \varphi_\theta(t,\xi_0)}_{\mathcal{V}_{\eta}} \!\leq\! \norm{\varphi_\theta(t,\xi_0)}_{\mathcal{V}_{\theta}}  \!\leq\! \overline{\kappa} \norm{T_\eta \varphi_\theta(t,\xi_0)}_{\mathcal{V}_{\eta}}
\end{align}
By \cite[Theorem 1]{IG2025}, it follows that the left inequality holds with $\underline{\kappa} = \norm{T_{\eta}}^{-1}$. Next, let 
    $\sigma_{\mathbbl{1}}\coloneqq \min\nolimits_{\theta \perp \mathbbl{1}_n, \norm{\theta}=1} \norm{VB^\mathsf{T}\theta} \in \mathbb{R}_{>0}$.
Then, the upper bound of \eqref{eq:normbound} holds with $\overline{\kappa} = \frac{1}{\min\{1, \sigma_{\mathbbl{1}}\}}.$ By Proposition~\ref{prop:nodal_exp_stability}, \eqref{eq:primaldualedge} is semi-globally exponentially stable on $\mathbb{R}^{e} \times \mathbb{R}^{2n}_{\geq0}$ with respect to $\mathcal{V}_{\eta}$. Using Proposition~\ref{prop:nodal_exp_stability} and the upper bound of \eqref{eq:normbound}, we have
\begin{align*}
    \norm{\varphi_\eta(t,T_\eta \xi_0)}_{\mathcal{V}_{\eta}} \leq \mathcal{M}_{\beta} \cdot e^{-\beta t}\norm{\varphi_\eta(0,T_\eta \xi_0)}_{\mathcal{V}_{\eta}}
\end{align*}
where 
    $\norm{\varphi_\theta(t,\xi_0)}_{\mathcal{V}_{\theta}} \leq \overline{\kappa}  \norm{\varphi_\eta(t, T_\eta \xi_0)}_{\mathcal{V}_{\eta}}$ 
and $\mathcal{M}_{\beta}$ is a constant depending on the decaying rate $\beta$. Moreover, we have $\norm{\varphi_\eta(0, T_\eta \xi_0)}_{\mathcal{V}_{\eta}} \leq \underline{\kappa}^{-1} \norm{\varphi_\theta(0, \xi_0)}_{\mathcal{V}_{\theta}}$. Finally, we conclude
\begin{align*}
    \norm{\varphi_\theta(t,\xi_0)}_{\mathcal{V}_{\theta}} \leq \overline{\kappa}\cdot\underline{\kappa}^{-1} \cdot \mathcal{M}_{\beta} \cdot e^{-\beta t}\norm{\varphi_\theta(0, \xi_0)}_{\mathcal{V}_{\theta}}.
\end{align*}
In other words, \eqref{eq:newplimdroopsc} is semi-globally exponentially stable on $\mathbb{R}^n \times \mathbb{R}^{2n}_{\geq 0}$ with respect to the set $\mathcal{V}_{\theta}$.
In addition, we show that $\lim_{t \to \infty} \omega(t) = \mathbbl{1}_n \omega_{s}$. According to Proposition~\ref{prop:nodal_exp_stability}, any pair $\left(\eta^\star, \lambda^\star\right)$ converges to a KKT point $\left(\eta^\star, \lambda^\star \right) \in \mathcal{V}_{\eta}$, i.e., $\lim_{t \to \infty} \eta(t) = \eta^\star$ and $\lim_{t \to \infty} \ddt \eta=\mathbbl{0}_e$. Using $\eta = V B^\mathsf{T} \theta$, we obtain
\begin{align*}
   \lim_{t \to \infty} \ddt V B^\mathsf{T} \theta(t) = \lim_{t \to \infty} V B^\mathsf{T} \ddt \theta(t) = \lim_{t \to \infty} VB^\mathsf{T} \omega(t) = 0
\end{align*}
and $\lim_{t \to \infty} \omega(t) = \mathbbl{1}_n \omega_{s}$ follows $\ker(B^\mathsf{T})=\vspan(\mathbbl{1}_n)$ to conclude the proof.
\end{proof}

The following corollary directly follows from \cite[Theorem 1]{TQL2020} and the proof of Theorem~\ref{thm:nodal_exp_stability}.
\begin{corollary}[\textbf{Semi-global exponential convergence}]\label{corr:rate}
Let $\beta \in \mathbb{R}_{>0}$ denote any strictly positive constant satisfying \eqref{eq:beta_bound}. Then, there exists $\mathcal{M}_{\beta} \in \mathbb{R}_{>0}$ such that $\lim_{\beta \to 0^+} \mathcal{M}_{\beta} = 1$ and
\begin{align*}
    \norm{\varphi_\theta(t,\xi_0)}_{\mathcal{V}_{\theta}} \leq \overline{\kappa}\cdot\underline{\kappa}^{-1} \cdot \mathcal{M}_{\beta} \cdot e^{-\beta t}\norm{\varphi_\theta(0, \xi_0)}_{\mathcal{V}_{\theta}}.
\end{align*}
\end{corollary}

\begin{proof}[Proof of Theorem~\ref{thm:controlgainconvergence}]
        Let $\zeta_{\beta} = \frac{\alpha}{\beta} - \frac{16\beta^2}{\kappa}, \zeta_{1} = \frac{4 \mathcal{F}^2}{\kappa}$, and $\zeta_2 = \frac{4}{\kappa}\big((\gamma + \mathcal{M})(\alpha + \mathcal{M} + \frac{1}{\rho}) + \frac{1}{2\rho^2}\big)$. Under the conditions of Corollary~\ref{corr:rate}, it holds that $\zeta_{\beta} \geq \zeta_{1} + \zeta_{2} + 1$. Notably, $\zeta_\beta$ is an increasing function with respect to $\kappa$, i.e., $\frac{\partial \zeta_{\beta}}{\partial \kappa} = \frac{16\beta^2}{\kappa^2} > 0$. Scaling the gains $k_i^{\prime} = s k_i$, we obtain $\mathcal{F}^\prime  = \sqrt{s_{k}} \mathcal{F}, \kappa^\prime =  s \kappa$ and $\mathcal{S} ^\prime = \sqrt{s} \mathcal{S}$ respectively. In other words, $\zeta_{\beta}$ increases and $\frac{\partial \zeta_1}{\partial s} = \frac{\partial}{\partial s} \frac{4{\mathcal{F}^\prime}^2}{{\kappa^\prime}^2} = 0$. Therefore, it is sufficient to show that $\zeta_2$ decreases under scaling by $s$. Expanding $\zeta_{2}$ results in 
        \begin{align*}
            \zeta_{2} = 4\bigg({\frac{\gamma \alpha}{\kappa}} + {\frac{(\gamma + \alpha) \mathcal{M}}{\kappa}} + {\frac{\gamma}{\kappa \rho}} + {\frac{\mathcal{M}(\mathcal{M} + \frac{1}{\rho})}{\kappa}} + {\frac{1}{2\kappa\rho^2}}\bigg)\!.
        \end{align*}
        Note that $\frac{\partial}{\partial s} \frac{\gamma \alpha}{\kappa^\prime} < 0$ and $\frac{\partial}{\partial s} \frac{\gamma}{\kappa^\prime \rho^\prime} < 0$, i.e., the first and third term decrease as $s$ increases. Moreover, the last term is unchanged since $\frac{\partial}{\partial s} \frac{1}{\kappa^\prime {\rho^\prime}^2} = 0$.
        Next, we show that the second term of $\zeta_2$ decreases and the fourth term of $\zeta_2$ is unchanged as $s$ increases. For brevity consider $\mathcal{M}^{\prime} = \mathcal{M}(\rho^\prime, \mathcal{F}^{\prime}, \mathcal{S}^{\prime})$, it results in
        \begin{align*}
            \frac{\partial}{\partial s} \!\frac{\mathcal{M}^{\prime}}{\kappa^\prime} \!=\! \frac{\partial}{\partial s}\frac{\rho(\mathcal{F}^2 + \mathcal{F} \mathcal{S} \nu_{0})}{\sqrt{s} \kappa} + \!
            \frac{\partial}{\partial s} \frac{\mathcal{S}(\nu_{0} + \|\mu^\star\|)}{\sqrt{s}\kappa} < \!0.
        \end{align*}
        Additionally, for the fourth term
        \begin{align*}
            \frac{\partial}{\partial s}\frac{{\mathcal{M}^{\prime}}^2}{\kappa^\prime} &= \frac{\partial}{\partial s}\frac{{\rho}^2(\mathcal{F}^2 + \mathcal{S} \mathcal{F} \nu_{0})^2}{\kappa} + \frac{\partial}{\partial s}\frac{\mathcal{S}^2(\nu_{0}+\|\mu^\star\|)^2}{ \kappa}\\
            &+ \frac{\partial}{\partial s}\frac{2\rho(\mathcal{F}^2 + \mathcal{S} \mathcal{F} \nu_{0})\mathcal{S}(\nu_{0} + \|\mu^\star\|)}{\kappa} = 0
        \end{align*}
        and
        \begin{align*}
            \frac{\partial}{\partial s}\frac{\mathcal{M}^{\prime}}{\kappa^\prime \rho^\prime} &= \frac{\partial}{\partial s} \frac{\rho \mathcal{F}^2 + \mathcal{S}(\rho \mathcal{F} \nu_{0} + \nu_{0} + \|\mu^\star\|)}{\kappa\rho} \\
            &= \frac{\partial}{\partial s} \frac{\mathcal{F}^2}{\kappa} + \frac{\partial}{\partial s} \frac{\mathcal{S} \mathcal{F} \nu_{0}}{\kappa} + \frac{\partial}{\partial s} \frac{\mathcal{S}(\nu_{0} + \|\mu^\star\|)}{\kappa \rho} = 0
        \end{align*}
    
       We conclude that $\frac{\kappa\alpha}{4\beta} - 4\beta^2$ increases for any increasing $s > 1$, i.e., $\frac{\partial \zeta_2}{\partial s} < 0 $. Moreover, recalling $\beta \leq \frac{\kappa\delta_{\min}}{46\rho \mathcal{F}^2}$  and the definition of $\delta_{\min}$ we have 
        \begin{align*}
            \delta_{\min}(\rho, k_i) &\coloneq 1 - \bigg[1+\rho \cdot \frac{\max\limits_{i \in \mc N \setminus (\mathcal{I}_u\cup \mathcal{I_\ell})}{\sqrt{k_i} g_i(\eta^\star)}}{\nu_{0}}\bigg]_{+}^2.
        \end{align*}
        This implies that $\frac{\partial}{\partial s} \frac{\kappa^\prime \delta_{\min}(\rho^\prime, k_i^\prime)}{46\rho^\prime {\mathcal{F}^\prime}^2} = \frac{\partial}{\partial s} \frac{\sqrt{s}\kappa \delta_{\min}(\rho, k_i)}{46\rho \mathcal{F}^2} > 0$.
        Due to enlargement of both bounds on $\beta$ in Corollary~\ref{corr:rate}, irrespective of the dependency of $M_{\beta}$ on $\beta$, we conclude that $\beta(k_i, \rho) < \beta(k^\prime_{i}, \rho^\prime)$.
\end{proof}

\begin{proof}[Proof of Proposition~\ref{prop:equivalency}]
Using the change of variables $\mu = K_I \lambda$ and \cite[Lemma 2]{IG2025}, the networked dynamics \eqref{eq:vector_net_nodal_orig} can be written as
\begin{subequations}\label{eq:vector_net_nodal}
    \begin{align}
        \ddt \theta =& M(P^\star  -L\theta - P_{L})\\ &- (\Xi \otimes K_I) \Pi_{\mathbb{R}^{2n}_{\geq 0}}\left(\rho (I_2 \otimes K_{I}) g(L\theta) + \mu\right)\nonumber\\
        \rho \ddt \mu =& \Pi_{\mathbb{R}^{2n}_{\geq 0}}\big(\rho (I_2 \otimes K_I) g(L\theta) + \mu\big) - \mu
   \end{align}
\end{subequations}
Recalling the definition $g(P_N)$ and using edge coordinate $\eta = VB^\mathsf{T}\theta$, the network dynamics in edge coordinates become identical to \eqref{eq:primaldualedge}.
Moreover, since $\ddt \eta  \in \Ima VB^\mathsf{T}$ holds for \eqref{eq:ddteta} and any initial condition $\eta_0 \in \Ima VB^\mathsf{T}$, it holds that $\eta \in \Ima VB^\mathsf{T}$ for all $t\in \mathbb{R}_{\geq 0}$. For any initial condition $(\theta_0, \mu_0)$ trajectories of \eqref{eq:vector_net_nodal} mapped to the edge coordinates coincide with trajectories of \eqref{eq:primaldualedge} with the initial condition $(VB^\mathsf{T}\theta_0, \mu_0)$. We conclude that the projection-free networked dynamics \eqref{eq:newplimdroopsc} mapped to edge coordinate coincide with the primal-dual dynamics associated with \eqref{eq:auglagrangian} for any initial condition $\eta_{0} \in \Ima(VB^\mathsf{T})$.
\end{proof}

\begin{proof}[Proof of Lemma~\ref{lem:laplacian_eigenvalues}]
Since $\mathbbl{1}_{n} \in \ker(L)$, the sum of absolute value of off-diagonal terms of each row of the Laplacian $L$ is equal to its corresponding diagonal term. Thus Gershgorin's circle theorem implies that the eigenvalues of the Laplacian are located in a union of the closed discs, i.e., $\lambda_i \in \bigcup_{i = 1}^{n} \mathbbl{D}(\sum_{k} w_{(i,k)}, \sum_{k} w_{(i,k)})$. Next, let $w_{\Sigma, i} = \sum_{k} w_{(i,k)}$ and note that
 $ \mathbbl{D}(w_{\Sigma, i}, w_{\Sigma, i}) \subseteq \mathbbl{D}(\max\limits_{i}(w_{\Sigma, i}), \max\nolimits_{i}(w_{\Sigma, i})$.
Thus, using $\max\nolimits_{i}(\sum_{k} w_{(i,k)}) \leq w_{\max} d_{\max}$, the upper bound is obtained. To establish the lower bound, note that $\lambda_{\max} (L) \geq w_{\min}\lambda_{\max}(B^\mathsf{T}B)$. It suffices to show that $\lambda_{\max}(B^\mathsf{T}B) \geq 1+d_{\max}$. To this end, by the Rayleigh quotient, we have
\begin{align*}
    \lambda_{\max}(B^\mathsf{T}B) \geq x^\mathsf{T} B^\mathsf{T} B x \quad \forall \|x\|=1
\end{align*}   
where $x^\mathsf{T} B^\mathsf{T} B x = \sum\nolimits_{(i,j) \in \mathcal{E}}(x_i - x_j)^2$. Next, let $\Delta \coloneqq  \frac{1}{\sqrt{d_{\max}^2 + d_{\max}}}$ and $x= \Delta  \begin{bmatrix}
  d_{\max} &
  -1 &
  \dots&
  -1 &
  0 & 
    \dots&
    0
\end{bmatrix}$.
The Proposition follows by noting that $x$ has $d_{\max}$ number of elements equal to $-\Delta$ and $n - d_{\max} - 1 $ of zero elements.
\end{proof}
\begin{proof}[Proof of Lemma~\ref{lem:mubound}]At any KKT point $\omega_i=\omega_s$ and either $\mc I_u=\emptyset$ or $\mc I_\ell=\emptyset$~\cite{IG2025}. Without loss of generality assume that $\mc I_\ell=\emptyset$ and $\lambda^\star_\ell=\mathbbl{0}_n$. Then, for all $i \in \mc N$, it holds that
    \begin{align}\label{eq:lambdabound_proj}
    \omega_s =& m_i (P^\star_i-P_i)-k_i \Pi_{{\mathbb{R}}_{\geq 0}}(\rho(P_i-P_{u,i}) + \lambda^\star_{u, i}) 
    \end{align}
Moreover, by primal feasibility any KKT point must satisfy $P_i-P_{u,i} \leq 0$ for all $i \in \mc N$ and, by complementary slackness, $\lambda^\star_{u,i}=0$ if $i \notin \mc I_u$. Thus, we either have $\Pi_{{\mathbb{R}}_{\geq 0}}(\rho(P_i-P_{u,i}) + \lambda^\star_{u, i})=0$ if $i \notin \mc I_u$ or $\Pi_{{\mathbb{R}}_{\geq 0}}(\rho(P_i-P_{u,i}) + \lambda^\star_{u, i})=\lambda^\star_{u, i}$ if $i \in \mc I_u$. Thus, for all $i \in \mc N$, \eqref{eq:lambdabound_proj} reduces to $\omega_s = m_i (P^\star_i-P_i) - k_i \lambda^\star_{u, i}$ and $\frac{k_i}{m_i} \lambda^\star_{u, i} = P^\star_i-P_i - \frac{\omega_s}{m_i}$. Summing over $i \in \mc N$ results in $\sum_{i=1}^n \frac{k_i}{m_i}  \lambda^\star_{u, i} = \sum_{i=1}^n  P^\star_i - \sum_{i=1}^n  P_i - \sum_{i=1}^n  \frac{\omega_s}{m_i}$. Because $P=L\theta+P_L$, it holds that $\sum_{i=1}^n  P_i= \sum_{i=1}^n  P_{L,i}$. Moreover, $\mu^\star_i = \sqrt{k_i} \lambda^\star$ and we obtain $\sum_{i=1}^n \frac{\sqrt{k_i}}{m_i}  \mu^\star_{u, i} = \sum_{i=1}^n  (P^\star_i - P_{L,i}) - \sum_{i=1}^n  \frac{\omega_s}{m_i} = \varrho$. Next, note that $\varrho \in \mathbb{R}_{\geq 0}$ must hold. Thus, it follows that $\frac{\sqrt{k_{\min}}}{m_{\max}}   \sum_{i=1}^n \mu^\star_{u, i} \leq \sum_{i=1}^n \frac{\sqrt{k_i}}{m_i}  \mu^\star_{u, i} = \varrho$. Moroever, because $\mu^\star_{u,i} \geq 0$ it immediately follows that 
$\norm{\mu^\star_{u}} \leq \norm{\mu^\star_{u}}_1 = \sum_{i=1}^n \mu^\star_{u, i} \leq \frac{m_{\max}}{\sqrt{k_{\min}}} \varrho$. The proof for $\mc I_u=\emptyset$ follows from the same steps.
\end{proof} 
\begin{proof}[Proof of Proposition~\ref{prop:rho_bound}] 
    Since $\frac{\kappa \alpha}{4\beta} - 4\beta^2$ is a decreasing function of $\beta$, any $\rho$ satisfying the conditions of Corollary~\ref{corr:rate} also satisfies
     $   \frac{23}{2}\kappa \alpha \rho^3 - c_1 \rho^2 - \gamma \rho - \frac{531}{1058} > 0.$
    where $c_1 = \mathcal{F}^2 + \frac{\kappa}{4} + \gamma \alpha$. Therefore, since $\gamma \rho + \frac{531}{1058} > 0$ we have $\frac{23}{2}\kappa \alpha \rho^3 > c_1 \rho^2$ and $\rho > \frac{2c_1}{23\kappa \alpha}$. Moreover, an upper bound on the solution is given by 
     $   \rho \leq 1 + \frac{1}{23\kappa\alpha}\max\{2c_1, 2\gamma, \frac{1062}{1058}\}.$
    The result follows from feasibility of the optimal $\rho^\star$.
    \end{proof}
\begin{proof}[Proof of Proposition~\ref{prop:connectivity}]
    Adding an edge between nodes $(i,j)$ increases the connectivity $\lambda_{\min}^+(L)$. In turn, this results in increased $\alpha$. Let $f(\alpha, \beta) \coloneqq \frac{\kappa\alpha}{4\beta} - 4\beta^2 $, and $f_2 = \mathcal{F}^2 + \frac{\kappa}{4} + (\gamma + \mathcal{M})(\alpha + \mathcal{M} + \frac{1}{\rho}) + \frac{1}{2\rho^2}$. Then, \eqref{eq:beta_cubic} can be written as $f(\alpha, \beta) \geq f_2$.  Moreover, $f(\alpha, \beta) \geq f_2$ is satisfied for all $\beta \leq \min\{\frac{\kappa \delta_{\min}}{46\mathcal{F}^2}, \arg_{\beta} (f(\alpha, \beta) = f_2)\}$.
    Taking the derivative of $f(\alpha, \beta) \geq f_2$ with respect to $\alpha$ and rearranging the resulting inequality results in $\beta \leq \frac{\kappa}{4 (\gamma + \mathcal{M})}$. If \eqref{eq:beta_constant} is binding, it is sufficient to have $\rho^\prime \geq \rho$. Otherwise,  \eqref{eq:beta_cubic} is binding if $\frac{\kappa}{4 (\gamma + \mathcal{M})} \leq \frac{\kappa \delta_{\min}}{46\rho\mathcal{F}^2}$. Therefore it results in $46\rho \mathcal{F}^2 \leq \delta_{\min} 4 (\gamma + \mathcal{M})$. Therefore it is sufficient to have
        \begin{align*}
            \tfrac{46}{4}\rho\mathcal{F}^2 \leq \gamma + \rho \mathcal{F}^2 + \mathcal{S}\left(\rho \mathcal{F} \nu_{0} + \nu_{0} + \tfrac{m_{\max}}{\sqrt{k_{\min}}}\varrho\right).
        \end{align*}
    For this inequality to hold, $\rho$ has to satisfy 
        \begin{align*}
            \rho\left(1-\frac{2}{21}\frac{\mathcal{S}\nu_{0}}{\mathcal{F}}\right) \leq \tfrac{2}{21}\left(\frac{\gamma}{\mathcal{F}^2} + \frac{\mathcal{S}\nu_{0}}{\mathcal{F}^2} + \frac{m_{\max}}{\sqrt{k_{\min}}\mathcal{F}^2}\varrho\right)
        \end{align*}
        In addition, $\frac{\gamma}{\mathcal{F}^2} \leq \frac{m_{\max} w_{\max}}{k_{\max} w_\Sigma}$, $\frac{\mathcal{S}}{\mathcal{F}^2} = \frac{\sqrt{ w_{\max}}}{\sqrt{k_{\max}d_{\max}} w_\Sigma}$, and $\frac{\mathcal{S}}{\mathcal{F}}\nu_{0} = \sqrt{\frac{w_{\max}}{w_\Sigma}}\nu_{0}$. Therefore for any $\rho$ satisfying the above inequality, the $f(\alpha, \beta) - f_2$ is an increasing function of $\alpha$.
\end{proof}

\begin{proof}[Proof of Proposition~\ref{prop:nodal_exp_stability}]
     We begin by noting that $M \in \mathbb{S}^n_{\succ 0}$. Then, by \cite[Observation~7.1.8]{HJ2013}, $B^\mathsf{T} M B \in \mathbb{S}^{e \times e}_{\succ 0}$ if and only if $\rank{B}=e$. If $\mc G$ is a connected tree, then $n=e+1$ and by \cite[Lem. 9.2]{LNS}, $\rank{B}=e$. Conversely, if $\mc G$ contains cycles, then $e \geq n$ and $\rank{B} \leq e-1$. Thus, if $\mc G$ is a tree, then the cost function of \eqref{eq:pfangledifference} is strongly convex and $\mathcal{V}_{\eta}$ is a singleton. Moreover, by \cite[Prop.~3]{IG2025} there exists $\eta$ such that $P_\ell < B V \eta + P_L < P_u$, i.e., Slater's condition holds. Then, by Assumption~\ref{assum:LICQ}, \cite[Thm.~1]{TQL2020} immediately implies that \eqref{eq:primaldualedge} is semi-globally exponentially stable with respect to $\mathcal{V}_{\eta}$. 
        
When $\mc G$ contains cycles, we can decompose \eqref{eq:pfangledifference} into a strongly convex part and convex part. Similarly, the dynamics \eqref{eq:primaldualedge} can be decomposed into an semi-globally exponentially stable part and Lyapunov stable part. To this end, let $\Gamma \coloneqq \begin{bmatrix} \Gamma_{+} & \Gamma_{0} \end{bmatrix}$ where $\Gamma_{+}  \in \mathbb{R}^{e \times n-1}$ contains eigenvectors corresponding to the positive eigenvalues of $V B^\mathsf{T} M B V$ and $\Gamma_{0}  \in \mathbb{R}^{e \times e-(n-1)}$ contains the eigenvectors corresponding to the zero eigenvalues. Next, let $\gamma = (\gamma_+,\gamma_0) \in \mathbb{R}^{e}$. Since $B^\mathsf{T} M B \in \mathbb{S}^n_{\succeq 0}$, we conclude that $\Gamma^{-1} =  \Gamma^\mathsf{T}$. Applying the change of coordinates $\eta = \Gamma \gamma$ to \eqref{eq:pfangledifference} results in
\begin{subequations}\label{eq:pfangledifferenceplus}
    \begin{align}
        &\min_{\gamma_+} \tfrac{1}{2} \norm{\gamma_+}^2_{H} + c^\mathsf{T} \gamma_+ 
        \label{eq:pfangledifferenceplus:obj}\\ 
        & \text{s.t. }   K_I P_{\ell} \leq K_I (A \gamma_+ + P_L)  \leq K_I  P_u,
        \label{eq:pfangledifferenceplus:constraint}
    \end{align}
\end{subequations}
where $H\coloneqq\Gamma_+^\mathsf{T} VB^\mathsf{T}MBV \Gamma_+$, $c\coloneqq\Gamma_+^\mathsf{T}VB^{\mathsf{T}}M(P_L\!-\!P^\star)$,  and $A\coloneqq BV \Gamma_+$. Notably, this transformation only removed redundant degrees of freedom and, by construction, \eqref{eq:pfangledifferenceplus} is strongly convex and strictly feasible under the same conditions as \eqref{eq:pfangledifference}. Moreover, since the transformation $\Gamma$ is invertible, by the invariant subspace principle, optimizers of \eqref{eq:pfangledifferenceplus} inherit LICQ from optimizers of \eqref{eq:pfangledifference}.

Moreover, given a KKT point $(\gamma^\star_+,\lambda^\star)$ of \eqref{eq:pfangledifferenceplus}, $BV \Gamma_0 \in \mathbb{R}^{n \times e-(n-1)}$ implies that $(\Gamma_+ \gamma^\star_+ + \Gamma_0 \gamma_0,\lambda^\star) \in \mathcal{V}_{\eta}$ for all $\gamma_0 \in \mathbb{R}^{e-(n-1)}$. Applying the change of coordinates $\eta = \Gamma \gamma$ to \eqref{eq:primaldualedge} results in $\ddt \gamma_0 =0$ and 
\begin{subequations}\label{eq:primaldualedgeplus}
\begin{align}
    \!\!\ddt \gamma_+=&-\!H \gamma_+ \!-\!c \!-\! (K_{I}A)^\mathsf{T} \big(\Pi_{\mathbb{R}^{n}_{\geq 0}}(\rho K_{I}g_{1}(A \gamma_+) + \mu_{\ell}) \nonumber\\
    & +  \Pi_{\mathbb{R}^{n}_{\geq 0}}(\rho K_{I}g_{2}(A \gamma_+) + \mu_{u}) \big),  \\ 
    \ddt \mu_{\ell} =& \frac{1}{\rho} \big(\Pi_{\mathbb{R}^n_{\geq 0} }\left(\rho K_I g_\ell(A\gamma_+) + \mu_{\ell}\right) - \mu_{\ell}\big),\\
    \ddt \mu_{u} =& \frac{1}{\rho} \big(\Pi_{\mathbb{R}^n_{\geq 0} }\left(\rho K_I g_u(A\gamma_+) + \mu_{u}\right) - \mu_{u}\big).
\end{align}
\end{subequations}
Notably, \eqref{eq:primaldualedgeplus} corresponds to primal-dual dynamics associated with the augmented Lagrangian of \eqref{eq:pfangledifferenceplus}. Thus, by \cite[Theorem 1]{TQL2020}, the dynamics \eqref{eq:primaldualedgeplus} are semi-globally exponentially stable with respect to a KKT point $(\gamma^\star_+,\lambda^\star)$ of \eqref{eq:pfangledifferenceplus}. In other words, \eqref{eq:primaldualedge} can be decomposed into dynamics that are semi-globally exponentially stable with respect to $(\gamma^\star_+,\lambda^\star)$ and a constant $\gamma_0 \in \mathbb{R}^{e \times e-(n-1)}$. Since $(\eta,\lambda) = (\Gamma_+ \gamma_+ + \Gamma_0 \gamma_0,\lambda) \in \mathcal{V}_{\eta}$ for any $\gamma_0$, it follows that \eqref{eq:primaldualedge} is semi-globally exponentially stable with respect to $\mathcal{V}_{\eta}$. The last statement of the Theorem follows by noting that $\ddt (\gamma_+,\lambda)=\mathbbl{0}_{3n-1}$ when $(\gamma_+,\lambda)=(\gamma^\star_+,\lambda^\star)$ and $\ddt \gamma_0=0$.
Moreover, (i) there exist upper bounds on the Frobenius and spectral norm of the Jacobian matrix of the constraints quantified in Lemma~\ref{prop:MgLg}, and (ii) when $\mathcal{G}$ is a tree, \eqref{eq:pfangledifference} is strongly convex, and if $\mathcal{G}$ is not a tree, then \eqref{eq:pfangledifferenceplus} is strongly convex. The proof is concluded by using the bounds in Lemma~\ref{lem:laplacian_eigenvalues} and Lemma~\ref{prop:MgLg} to establish the various constants used in the convergence bound in \cite[Theorem 1]{TQL2020}.
\end{proof}
\begin{lemma}\label{prop.lipschitz}
    The function $\tfrac{1}{2} \norm{\gamma_+}^2_{H} + c^\mathsf{T} \gamma_+$ is $\alpha$-strongly convex and its gradient is $\gamma$-Lipschitz. Let $\underline{\alpha} =  m_{\min} \lambda_{\min}^{+}(L)$ and $\overline{\gamma} = 2w_{\max} m_{\max}d_{\max}$, then it holds that $\underline{\alpha} \leq \alpha \leq \gamma \leq \overline{\gamma}$.
\end{lemma}
    \begin{proof}
        The the function $\tfrac{1}{2} \norm{\gamma_+}^2_{H} + c^\mathsf{T} \gamma_+$ is $\alpha$-strongly convex with $\alpha = \lambda_{\min}(H)$ and its gradient is $\gamma$-Lipschitz with $\gamma = \lambda_{\max}(H)$. Moreover, by construction of $H$ and $\Gamma_+$, it holds that $\alpha = \lambda^{+}_{\min}(V B^\mathsf{T} M BV)$, where $\lambda^{+}_{\min}$ denotes the smallest non-zero eigenvalue, and $\gamma = \lambda_{\max} (VB^\mathsf{T} M B V)$.

    Next, we define the edge Laplacian associated with the Laplacian $L$ as $L_{e} = V B^\mathsf{T} B V$. This results in
    \begin{align*}
 m_{\min} \lambda_{\max} (L_e) &\leq \lambda_{\max} (VB^\mathsf{T} M B V) \leq m_{\max} \lambda_{\max} (L_e).
 \end{align*}
  Using $\lambda_{\max}(L_{e}) = \lambda_{\max}(L)$~\cite{ZM2011} and Lemma~\ref{lem:laplacian_eigenvalues} results in
\begin{align*}
    m_{\min}w_{\min}(1+ d_{\max}) \leq \gamma \leq 2 m_{\max} w_{\max}d_{\max}.
\end{align*}
In addition,$ \alpha = \lambda^{+}_{\min}(V B^\mathsf{T} M BV) \geq m_{\min}\lambda^{+}_{\min}(L)$, where $\lambda_{\min}^{+}(L)$ is the algebraic connectivity of the graph $\mathcal{G}$. Thus, we conclude that\footnote{This bound holds for any graph $\mathcal{G}$ with $n\geq3$. Moreover, for $n = 2$ the second-smallest eigenvalue can be obtained explicitly.} $\underline{\alpha} \leq \alpha \leq \gamma \leq \overline{\gamma}$ holds for $\underline{\alpha} =  m_{\min} \lambda_{\min}^{+}(L)$ and $\overline{\gamma} = 2w_{\max} m_{\max}d_{\max}$.
    \end{proof}    

\begin{lemma}\label{prop:MgLg}
    Consider $k_{\max} = \max_{i \in \mathcal{N}} k_i$. Then, it holds that $\norm{K_I B V \Gamma_{+}}_{2} \leq \mathcal{S}$ and $\norm{K_I B V \Gamma_{+}}_{F} \leq \mathcal{F}$.
\end{lemma}
\begin{proof}
$\Gamma_{+}$ contains the set of eigenvectors corresponding to the positive eigenvalues of $V B^\mathsf{T} M B V$. Since the columns of $\Gamma_{+}$ are orthonormal, it holds that
$\norm{\Gamma_{+}}_{2} = 1$ and
\begin{align*}
    \norm{K_I B V \Gamma_{+}}_{2} &\leq \norm{K_I}_{2} \norm{BV}_{2} \norm{\Gamma_{+}}_{2}\leq \sqrt{k_{\max} \lambda_{\max}(L_{e})}\\
    &\leq \sqrt{k_{\max}\lambda_{\max}(L)} \leq \mathcal{S}.
\end{align*}
Moreover, it holds that $\norm{K_I B V \Gamma_{+}}_{F} \leq \norm{K_I}_{2} \norm{B}_{2} \norm{V}_{F} \norm{\Gamma_{+}}_{2}   \leq \sqrt{2 d_{\max} k_{\max} w_\Sigma} =  \mathcal{F}$, where we used $w_\Sigma\coloneqq\sum\nolimits_{j=1}^{e}w_j$, $\|B\|_2 = \sqrt{\|B\|_1 \|B\|_{\infty}}$, $\|B\|_{\infty} = d_{\max}$, and  $\|B\|_1 = 2$.
\end{proof}

\bibliographystyle{IEEEtran}
\bibliography{IEEEabrv,extension}

\end{document}

%% file: ieee9bus_table.tex
\begin{table}[H]
    \centering
    \caption{Model and control parameters. For further details see \cite[Table I]{TGA+2020}.}
    \label{table:parameters}
    \scriptsize
    \resizebox{\columnwidth}{!}{
    \begin{tabular}{|@{}c@{}|c|c|c|c|c|c|}
    \hline
    \multirow{2}{*}{VSC} & \multicolumn{3}{c|}{\cellcolor{lightgray}Power [MW]} & \multicolumn{3}{c|}{\cellcolor{lightgray} Control gains [pu]} \\
    \cline{2-7} 
    \rule{0pt}{2.1ex}
    & $P^\star$ & $P_\ell$ & $P_u$ & $m_p$ & $\rho$ & $k_I$ \\
    \hline
    1 & 25 MW & 20 MW & 110 MW & 4.17\% & 1.02 & 40.95 \\
    \hline
    2 & 87.5 MW & 20 MW & 110 MW & 9.38\% & 1.02 & 40.95 \\
    \hline
    3 & 55 MW & 20 MW & 110 MW & 6\% & 1.02 & 40.95\\
    \hline
    \end{tabular}
    }
\end{table}